\documentclass[12pt]{article}

\usepackage[utf8]{inputenc}
\usepackage[left=1in, right=1in, top=1in, bottom=1in]{geometry}
\usepackage{setspace}
\usepackage{amsmath}
\usepackage{amssymb}
\usepackage{amsthm}
\usepackage{graphicx}
\usepackage{tikz}
\usetikzlibrary{positioning, decorations.pathreplacing}
\usepackage{hyperref}
\usepackage{cite}
\usepackage{algorithm}
\usepackage{algpseudocode}
\usepackage{subcaption}

\newtheorem{theorem}{Theorem}
\newtheorem{lemma}[theorem]{Lemma}
\newtheorem{definition}{Definition}

\title{Cardinal Grid Slime Trail is PSPACE-Complete}
\author{Anne Pham \\ Matt Ferland \\[4pt] Department of Computer Science \\ Dickinson College \\ Carlisle, PA 17013}
\date{}

\begin{document}

\maketitle

\begin{abstract}
Slime Trail is a two-player combinatorial game in which the players alternately move a shared token to an adjacent vertex, permanently removing each vertex the token leaves, while attempting to reach a goal node. Ferland and Burke (2017) proved that Slime Trail is PSPACE-complete on arbitrary planar graphs and asked whether the same holds for the grid version actually used in play. We resolve this open problem by proving that \textsc{Cardinal Grid Slime Trail}, that is, Slime Trail on a square grid with four-directional movement, is PSPACE-complete. We adapt the QBF reduction of \cite{ferland2017} to the grid setting, designing grid-compatible gadgets that respect the degree-4 bound and the parity constraints of the integer lattice. We further show the construction extends, under a 45-degree rotation, to the eight-directional variant.
\end{abstract}

\section{Introduction}

Slime Trail is a two-player combinatorial game created by Bill Taylor in 1992. 
It has gained significant popularity, particularly in Portugal, where it has 
been featured in the annual National Mathematical Games Championship (CNJM) 
organized by Ludus since 2008. The competition currently draws approximately 
100,000 participants annually, with around 2,000 players competing in the 
championship finals \cite{ludus_cnjm}. The game is played on a grid where 
players alternate moving a token between squares, leaving a ``slime trail'' 
that makes previously visited vertices inaccessible.

The game belongs to the broader class of combinatorial games: two-player,
perfect information games in which players alternate moves with no random
elements \cite{albert2007}. These games have been extensively studied, with the field of combinatorial game theory dedicated to studying them \cite{fraenkel1996, berlekamp2001}. A detailed account of the rules of
Slime Trail, including the winning conditions and the reachability constraint
that governs legal moves, is provided in Section~\ref{sec:rules}.

As mentioned in \cite{ferland2017}, Slime Trail is 
PSPACE-complete when played on arbitrary planar graphs, via a reduction from 
the Quantified Boolean Formula (QBF) problem. However, their paper explicitly 
identified the grid version as an open problem: \textit{``Is Slime Trail still 
PSPACE-complete when played on a grid?''} Grid graphs impose stricter geometric 
constraints than arbitrary planar graphs, so it is not immediately obvious 
whether complexity results transfer between these two settings. This has 
particular relevance as the actual game of Slime Trail, as played in tournaments 
and practice, uses square grids rather than arbitrary graphs.

\subsection{Variant Studied in This Paper}

This paper studies \textit{Cardinal Grid Slime Trail}: Slime Trail played on 
a finite square grid with strictly four-directional (orthogonal) movement, 
where players may only move the token up, down, left, or right to an adjacent 
square. This differs from the original game described by Bill Taylor in 1992, 
which permits movement in all eight directions (including diagonals), and from 
the version played in Portuguese national tournaments organized by Ludus 
\cite{ludus_cnjm}, which also uses eight-directional movement. We note, however, that our construction is also compatible with the 
eight-directional variant with a simple transformation (see Section~\ref{sec:eight-directional} for a full argument).

\subsection{Our Contributions}

The primary contribution of this work is a proof that Cardinal Grid Slime Trail is PSPACE-complete. To achieve this, we develop grid-compatible gadget designs that preserve the logical behavior of their planar graph counterparts. We introduce a systematic approach to managing parity constraints through strategic dummy node placement and construct a complete polynomial-time reduction from QBF to Cardinal Grid Slime Trail.

\section{Game Rules}\label{sec:rules}

\subsection{Overview}

Slime Trail is played on a grid by two players, referred to as Blue and Red.
The board is a rectangular grid of squares, and two or more squares are designated as
goal squares: at least one belonging to Blue and at least one belonging to Red. In competition
play as organized by Ludus \cite{ludus_cnjm}, the goals are placed in
diagonally opposite corners of the board, giving both players an equally long
path to navigate \cite{mancala_slimetrail}. A single token is placed on any
mutually agreed starting square, and play begins with Blue moving first
\cite{ferland2017}. Figure~\ref{fig:game-example} shows one such instance of the Slime Trail game with 2 target squares and a starting position; in general, the board may contain more than two goal squares and the token may begin at any mutually agreed position.

\begin{figure}[htbp]
\centering
\begin{tikzpicture}[scale=1.0]

  \definecolor{boardyellow}{RGB}{255,193,37}
  \definecolor{slimegray}{RGB}{160,160,160}
  \definecolor{slimegreen}{RGB}{120,200,80}
  \definecolor{tokenblack}{RGB}{30,30,30}

  \foreach \col in {0,...,5} {
    \foreach \row in {0,...,4} {
      \fill[boardyellow] (\col,\row) rectangle (\col+1,\row+1);
    }
  }

  \fill[slimegray] (3,2) rectangle (4,3);  
  \fill[slimegray] (4,2) rectangle (5,3);  

  \fill[blue!80!black]  (5,4) rectangle (6,5);   
  \fill[red!80!black]   (0,0) rectangle (1,1);   

  \draw[black, line width=0.8pt]
    \foreach \col in {0,...,6} { (\col,0) -- (\col,5) }
    \foreach \row in {0,...,5} { (0,\row) -- (6,\row) };

  \draw[black, line width=2pt] (0,0) rectangle (6,5);

  \fill[tokenblack] (2.5,2.5) circle (0.28);

  \foreach \col/\row in {
    0/1, 0/2, 0/3, 0/4,
    1/0, 1/1, 1/2, 1/3, 1/4,
    2/0, 2/1,      2/3, 2/4,
    3/0, 3/1,      3/3, 3/4,
    4/0, 4/1,      4/3, 4/4,
    5/0, 5/1, 5/2, 5/3
  } {
    \node[font=\tiny, text=black!55] at (\col+0.5,\row+0.5) {$(\col,\row)$};
  }

  \node[font=\tiny\bfseries, text=white] at (5.5,4.5) {\shortstack{$B$ \\ $(5,4)$}};
  \node[font=\tiny\bfseries, text=white] at (0.5,0.5) {\shortstack{$R$ \\ $(0,0)$}};
  \node[font=\tiny,           text=white] at (3.5,2.5) {$(3,2)$};
  \node[font=\tiny,           text=white] at (4.5,2.5) {$(4,2)$};

  \draw[->, thick, slimegreen, line width=1.5pt]
    (2.5,2.5) -- (3.5,2.5);   
  \draw[->, thick, slimegreen, line width=1.5pt]
    (3.5,2.5) -- (4.5,2.5);   
  \draw[->, thick, slimegreen, line width=1.5pt]
    (4.5,2.5) -- (4.5,3.5);   

\end{tikzpicture}
\caption{%
  An example position in a game of Cardinal Grid Slime Trail on a $6 \times 5$ board.
  Squares are identified by zero-indexed coordinates $(\text{col}, \text{row})$
  with $(0,0)$ at the bottom-left.
  Blue's goal is square $(5,4)$ (top-right, blue) and Red's goal is square
  $(0,0)$ (bottom-left, red).
  The token (black circle) sits at its starting position $(2,2)$.
  Green arrows show three moves already made: $(2,2) \to (3,2) \to (4,2) \to (4,3)$.
  Gray squares ($(3,2)$ and $(4,2)$) have been vacated and are permanently
  inaccessible (``slimed'').}
\label{fig:game-example}
\end{figure}
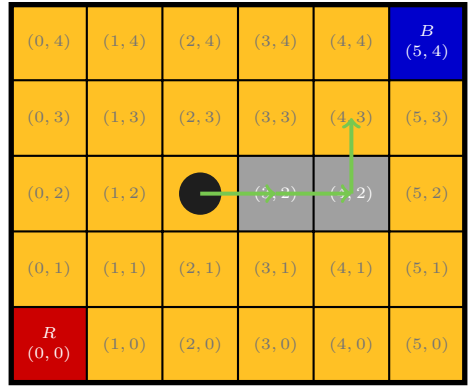

\subsection{Movement}

On each turn, the current player moves the token to an orthogonally adjacent
square, that is, a square that shares an edge horizontally or vertically with
the token's current square. The square the token has just vacated is
immediately marked as slimed and permanently removed from play; neither player
may enter a slimed square for the remainder of the game \cite{ferland2017,
burke2017blog}. This growing trail of inaccessible squares is the defining
mechanic of the game, and it is what gives Slime Trail its name.

\subsection{Winning and Losing}

A player wins
by moving the token onto their own goal square: if Blue moves the token to
Blue's goal, Blue wins immediately, and likewise for Red. Notably, either
player may be the one to move the token onto a goal square: if Blue moves the
token onto Red's goal, then Red wins. Second, a player wins if the opponent
has no legal move on the opponent's turn. 


\subsection{The Reachability Constraint}

A subtle but important rule governs which moves are legal. A player is not
permitted to move the token to a square from which neither goal square is
reachable by any path on the remaining board \cite{ferland2017, burke2017blog}.
This rule prevents a player from deliberately trapping the token in a dead
pocket of the board from which the game could never be resolved.

As noted in \cite{ferland2017}, it is possible for a player
to make a move that eliminates all paths to their own goal while leaving a path
to the opponent's goal. In this case the move is still legal, but the opponent
automatically wins, since the only reachable goal now belongs to them. The
reachability constraint therefore requires that at least one goal square remain
reachable after every move, not necessarily one's own.


\subsection{Rule Variants}

The exact rules for Slime Trail have varied over time and across communities.
The original ruleset described by Bill Taylor in 1992 permitted draws, allowing
the token to be moved to a square from which no goal was reachable
\cite{mancala_slimetrail}. The no-draw reachability constraint described above
was introduced later and is now the standard used in competition and in the
complexity-theoretic literature \cite{ferland2017}. The game is typically played with adjacency in all
eight directions rather than four \cite{burke2017blog}.

\section{Preliminaries}

\subsection{Slime Trail}

We begin by formally defining both the general and grid-specific versions of Slime Trail.

\begin{definition}[Slime Trail \cite{ferland2017}]
A \textit{starting position} for Slime Trail consists of a connected graph $G = (V, E)$ with at least one vertex colored blue, at least one vertex colored red, and one vertex containing a moveable token. Two players, Blue and Red, alternate turns, with Blue moving first. On each turn, the current player moves the token to an adjacent vertex and removes the previous vertex from the graph. A player wins by moving the token onto a vertex of their own color, or if the opponent has no legal moves. The token may never be moved to a position from which neither color's goal vertex is reachable.
\end{definition}

The game as played in competition, however, is not defined on arbitrary graphs but on structured grids. Grid graphs impose geometric constraints absent from the general setting: most notably a degree bound of eight and a bipartite structure that governs path parity. These constraints make direct translation of results from arbitrary planar graphs nontrivial, and they motivate a careful separate treatment of the grid case. We therefore define the relevant graph class precisely.

\begin{definition}[Grid Graph]
A \textit{grid graph} $G = (V, E)$ is a graph where each vertex $v \in V$ 
corresponds to a lattice point $(x, y)$ with $x, y \in \mathbb{Z}$, two 
vertices $(x_1, y_1)$ and $(x_2, y_2)$ are adjacent if and only if 
$\max(|x_1 - x_2|, |y_1 - y_2|) = 1$, and every vertex has degree at most 8.
\end{definition}

With this definition in hand, we can state precisely the variant of Slime Trail that is the subject of our main result.

\begin{definition}[Cardinal Grid Slime Trail]
\textit{Cardinal Grid Slime Trail} is Slime Trail played on a finite grid graph with strictly orthogonal (four-directional) movement, where all vertices, edges, and goal nodes must align to the integer lattice, and two vertices are adjacent only if they differ by exactly 1 in exactly one coordinate.
\end{definition}

\subsection{Grid Constraints}

Grid graphs impose several constraints not present in arbitrary planar graphs.

The most immediate difference is the restriction on maximum degree. While 
vertices in a general planar graph may have arbitrarily high degree, every 
vertex in a grid graph has degree at most 8, since each vertex has at most 
eight neighbors: four orthogonal and four diagonal. In the cardinal variant, 
degree is further restricted to at most 4.

A subtler but equally important constraint concerns path lengths. Because the 
grid is bipartite with respect to orthogonal moves, path lengths between 
vertices are bounded below by their Chebyshev distance, and the parity of a 
path's length determines which player holds the turn upon arrival. A path with 
an even number of edges means the player who began the path arrives at the 
endpoint; an odd length path means the opponent arrives.

This parity constraint means that one cannot insert intermediate vertices 
arbitrarily without affecting turn order and grid alignment. When translating 
a planar graph gadget to the grid, dummy nodes must be inserted in even 
quantities and in geometrically valid positions to preserve correct parity. These constraints motivate the careful gadget designs 
presented in Section~\ref{sec:gadgets}.

\subsection{Computational Complexity Background}

We briefly recall the complexity theoretic concepts underlying our hardness proof. The central class of interest is PSPACE, which captures problems solvable with polynomially bounded memory \cite{papadimitriou1994}, as defined in Definition~\ref{def:pspace}.

\begin{definition}[PSPACE \cite{papadimitriou1994}]\label{def:pspace}
A decision problem is in PSPACE if it can be solved by a deterministic Turing 
machine using space polynomial in the input size.
\end{definition}

Among problems in PSPACE, those that are hardest in the sense that every other PSPACE problem reduces to them are called PSPACE-complete \cite{papadimitriou1994}. We make this notion precise as in Definition~\ref{def:pspace-complete}.

\begin{definition}[PSPACE-complete \cite{papadimitriou1994}]\label{def:pspace-complete}
A problem is PSPACE-complete if it is in PSPACE and every problem in PSPACE 
reduces to it in polynomial time.
\end{definition}

The problem we reduce from in our main result is the Quantified Boolean Formula (QBF) problem, which is PSPACE-complete \cite{papadimitriou1994, sipser2012}. Intuitively, QBF (Definition ~\ref{def:qbf}) generalizes SAT by allowing alternating existential and universal quantifiers over Boolean variables.

\begin{definition}[Quantified Boolean Formula (QBF) \cite{papadimitriou1994}]\label{def:qbf}
The QBF problem asks whether a fully quantified Boolean formula of the form
$$\exists\, x_1 \; \forall\, x_2 \; \exists\, x_3 \; \forall\, x_4 \; \ldots \; Q_n\, x_n \; : \; 
\phi(x_1, x_2, \ldots, x_n)$$
is true, where $\phi$ is a Boolean formula in conjunctive normal form (CNF),
odd-indexed variables are existentially quantified, even-indexed variables 
are universally quantified, and $Q_n \in \{\exists, \forall\}$ depending on 
the parity of $n$.
\end{definition}

\section{Main Result}

We now state and prove our main theorem. The proof proceeds in two parts: membership in PSPACE, which follows from a straightforward space analysis, and PSPACE-hardness, which is established via a reduction from QBF in Section~\ref{sec:reduction}.

\begin{theorem}
Cardinal Grid Slime Trail is PSPACE-complete.
\end{theorem}

\begin{proof}
We establish this in two parts: Cardinal Grid Slime Trail is in PSPACE (Lemma \ref{lem:pspace}), and it is PSPACE-hard by reduction from QBF (Section \ref{sec:reduction}).
\end{proof}

The membership direction is handled by Lemma~\ref{lem:pspace}.

\begin{lemma}\label{lem:pspace}
Cardinal Grid Slime Trail is in PSPACE.
\end{lemma}

\begin{proof}
The argument follows the same reasoning as for planar Slime Trail 
\cite{ferland2017}. Let $N$ denote the total number of vertices in the grid 
graph. Since each move permanently slimes one vertex, every branch of the 
game tree has depth at most $N$. A depth-first search over the game tree 
requires $O(N)$ space to maintain the current path and $O(4N) = O(N)$ space 
to track which vertices and their neighbors have been slimed, for a total of 
$O(N)$ space. Since $N$ is polynomial in the input size, Cardinal Grid Slime Trail 
is in PSPACE.
\end{proof}

The harder direction, PSPACE-hardness, requires constructing a reduction from QBF. We describe this reduction in full in Section~\ref{sec:reduction}.

\section{The QBF Reduction}\label{sec:reduction}

To prove PSPACE-hardness, we reduce from QBF to Cardinal Grid Slime Trail. Since QBF is PSPACE-complete, showing that Cardinal Grid Slime Trail is at least as hard as QBF is sufficient to establish PSPACE-hardness. Given a QBF instance
$$\exists\, x_1 \; \forall\, x_2 \; \exists\, x_3 \; \forall\, x_4 \; \ldots \; Q_n\, x_n \; : \; \phi(x_1, x_2, \ldots, x_n)$$
where $\phi$ is in CNF with $m$ clauses, we construct a Cardinal Grid Slime Trail game in which Blue has a winning strategy if and only if the QBF formula is true.

\subsection{Overview of the Reduction}

The reduction follows the general structure introduced in \cite{ferland2017}, adapted for grid constraints. As illustrated in 
Figure~\ref{fig:reduction_workflow}, our construction consists of three main 
phases that mirror the structure of the QBF formula.

The overall structure of the reduction is also related to the framework used 
in Geography-style reductions \cite{schaefer1978, lichtenstein1980}, in which 
a token traverses a graph whose structure encodes a formula's quantifier 
alternation. Like Geography, our gadgets chain variable-setting moves before 
a clause-checking phase.

Before describing the phases, we establish the correspondence between players 
and quantifiers. Blue acts as the \textit{true} player, controlling existentially 
quantified variables: Blue wins by finding a satisfying assignment. Red acts as 
the \textit{false} player, controlling universally quantified variables: Red wins 
by finding a clause that Blue cannot satisfy. Concretely, Blue controls all 
odd-indexed variables $x_1, x_3, x_5, \ldots$ and Red controls all even-indexed 
variables $x_2, x_4, x_6, \ldots$, reflecting the alternating quantifier 
structure $\exists\, x_1 \; \forall\, x_2 \; \exists\, x_3 \; \forall\, x_4 \; \ldots$ of the QBF 
formula. This correspondence is central to the correctness of the reduction.

\subsubsection{Phase 1: Variable Setting}

The game begins with a sequence of variable-setting steps corresponding to the 
variables $x_1, x_2, \ldots, x_n$ in the QBF formula, chained together so 
that each step feeds into the next. Blue controls the existentially quantified 
variables and Red controls the universally quantified ones, alternating 
throughout. Each player's choice is permanent: the direction not taken becomes 
inaccessible for the rest of the game, creating a lasting record of variable 
assignments that will be checked in Phase 3. The grid gadgets that enforce this 
behavior are described in Section~\ref{sec:gadgets}.

After all $n$ variables have been set, the token reaches a node that connects 
to the clause selection phase, beginning Phase 2.

\subsubsection{Phase 2: Clause Selection}

Red adversarially selects one clause $C_i$ from the $m$ clauses in $\phi$, 
choosing the clause that Red determines Blue cannot satisfy given the variable 
assignments from Phase 1. Red has complete freedom in this selection, and Blue 
has no influence over which clause is tested. The grid gadget that 
enforces this adversarial choice is described in Section~\ref{sec:choice}.

\subsubsection{Phase 3: Literal Verification}

Within the selected clause, Blue must find at least one literal in $C_i$ that 
evaluates to true under the variable assignments from Phase 1. Blue does this 
by checking each literal in turn: if a literal is true, Blue can follow that 
path and win; if all literals are false, Blue has no winning move and loses. 
The grid gadgets that implement this verification are described in 
Section~\ref{sec:gadgets}.

\subsubsection{Connecting the Phases to QBF Semantics}

The three-phase structure directly encodes QBF evaluation. Phase 1 mirrors the 
quantifier prefix: Blue chooses values for existentially quantified variables 
while Red chooses values for universal ones. Phase 2 
represents universal quantification over clauses: Red's free clause selection 
means Blue cannot win unless Blue can satisfy every clause Red might choose. 
Phase 3 represents existential quantification over literals: Blue need only 
find one true literal in the selected clause.

Blue has a winning strategy if and only if Blue can choose assignments for the 
existentially quantified variables such that, no matter how Red sets the 
universally quantified variables and no matter which clause Red selects, Blue 
can find a true literal in that clause. This is precisely the condition for the 
QBF formula to be true.

\begin{figure}[htbp]
    \centering
    \begin{tikzpicture}[
    scale=0.8,
    transform shape,
    node distance=0.8cm,
    every node/.style={circle, draw, thick, minimum size=0.6cm, font=\small, fill=white}
]

\coordinate (a0) at (0,0);
\coordinate (a1) at (-0.9,-1.2);
\coordinate (a1bar) at (0.9,-1.2);
\coordinate (a2) at (0,-2.1);

\coordinate (b0) at (0,-3.3);
\coordinate (b1) at (-0.9,-4.5);
\coordinate (b1bar) at (0.9,-4.5);
\coordinate (b2) at (0,-5.4);

\coordinate (dots1) at (0,-6.6);

\coordinate (c0) at (0,-7.8);
\coordinate (c1) at (-0.9,-9);
\coordinate (c1bar) at (0.9,-9);
\coordinate (c2) at (0,-9.9);

\coordinate (dots2) at (0,-11.1);

\coordinate (n0) at (0,-12.3);
\coordinate (n1) at (-0.9,-13.5);
\coordinate (n1bar) at (0.9,-13.5);
\coordinate (n2) at (0,-14.4);

\coordinate (choice_entry) at (0,-15.6);

\coordinate (clause1) at (-4.5,-17.1);
\coordinate (clause2) at (-1.5,-17.1);
\coordinate (clause3) at (1.5,-17.1);
\coordinate (clause_dots) at (3.5,-17.1);
\coordinate (clause_m) at (5.5,-17.1);

\coordinate (lit_left) at (-4.5,-18.6);
\coordinate (lit_mid) at (-1.5,-18.6);
\coordinate (lit_3) at (1.5,-18.6);
\coordinate (lit_right) at (5.5,-18.6);

\coordinate (lit_left2) at (-4.5,-20.1);
\coordinate (lit_mid2) at (-1.5,-20.1);
\coordinate (lit_32) at (1.5,-20.1);
\coordinate (lit_right2) at (5.5,-20.1);

\coordinate (vdots_left) at (-4.5,-21.3);
\coordinate (vdots_mid) at (-1.5,-21.3);
\coordinate (vdots_3) at (1.5,-21.3);
\coordinate (vdots_right) at (5.5,-21.3);

\draw[thick] (a0) -- (a1);
\draw[thick] (a0) -- (a1bar);
\draw[thick] (a1) -- (a2);
\draw[thick] (a1bar) -- (a2);

\draw[thick] (a2) -- (b0);
\draw[thick] (b0) -- (b1);
\draw[thick] (b0) -- (b1bar);
\draw[thick] (b1) -- (b2);
\draw[thick] (b1bar) -- (b2);

\draw[thick, dotted, line width=2pt] (b2) -- (dots1);

\draw[thick] (dots1) -- (c0);
\draw[thick] (c0) -- (c1);
\draw[thick] (c0) -- (c1bar);
\draw[thick] (c1) -- (c2);
\draw[thick] (c1bar) -- (c2);

\draw[thick, dotted, line width=2pt] (c2) -- (dots2);

\draw[thick] (dots2) -- (n0);
\draw[thick] (n0) -- (n1);
\draw[thick] (n0) -- (n1bar);
\draw[thick] (n1) -- (n2);
\draw[thick] (n1bar) -- (n2);

\draw[thick] (n2) -- (choice_entry);

\draw[thick, ->] (choice_entry) -- (clause1);
\draw[thick, ->] (choice_entry) -- (clause2);
\draw[thick, ->] (choice_entry) -- (clause3);
\draw[thick, ->] (choice_entry) -- (clause_m);

\draw[thick] (clause1) -- (lit_left);
\draw[thick] (clause2) -- (lit_mid);
\draw[thick] (clause3) -- (lit_3);
\draw[thick] (clause_m) -- (lit_right);

\draw[thick, dashed] (lit_left) .. controls (-6,-14) and (-4.5,-7) .. (b2);
\draw[thick, dashed] (lit_mid) .. controls (-2.5,-15) and (-1.5,-11) .. (c2);
\draw[thick, dashed] (lit_3) .. controls (4,-20) and (4,-3) .. (a2);
\draw[thick, dashed] (lit_right) .. controls (7,-15) and (6,-14.5) .. (n2);

\draw[thick] (lit_left) -- (lit_left2);
\draw[thick] (lit_mid) -- (lit_mid2);
\draw[thick] (lit_3) -- (lit_32);
\draw[thick] (lit_right) -- (lit_right2);

\draw[thick, dashed] (lit_left2) .. controls (-7,-17) and (-5,-9) .. (a2);
\draw[thick, dashed] (lit_mid2) .. controls (-3.5,-18) and (-2,-13.5) .. (n2);
\draw[thick, dashed] (lit_32) .. controls (3.5,-22) and (3,-5) .. (b2);
\draw[thick, dashed] (lit_right2) .. controls (8,-18) and (7,-11) .. (c2);

\node at (a0) {$a_0$};
\node at (a1) {$a_1$};
\node at (a1bar) {$\bar{a}_1$};
\node at (a2) {$a_2$};

\node at (b0) {$b_0$};
\node at (b1) {$b_1$};
\node at (b1bar) {$\bar{b}_1$};
\node at (b2) {$b_2$};

\node[draw=none, font=\large\bfseries, fill=none] at (dots1) {$\vdots$};

\node at (c0) {$c_0$};
\node at (c1) {$c_1$};
\node at (c1bar) {$\bar{c}_1$};
\node at (c2) {$c_2$};

\node[draw=none, font=\large\bfseries, fill=none] at (dots2) {$\vdots$};

\node at (n0) {$n_0$};
\node at (n1) {$n_1$};
\node at (n1bar) {$\bar{n}_1$};
\node at (n2) {$n_2$};

\node at (choice_entry) {};

\node at (clause1) {$C_1$};
\node at (clause2) {$C_2$};
\node at (clause3) {$C_3$};

\node[draw=none, font=\large\bfseries, fill=none] at (clause_dots) {$\cdots$};
\node at (clause_m) {$C_m$};

\node at (lit_left) {};
\node at (lit_mid) {};
\node at (lit_3) {};
\node at (lit_right) {};

\node at (lit_left2) {};
\node at (lit_mid2) {};
\node at (lit_32) {};
\node at (lit_right2) {};

\node[draw=none, font=\large\bfseries, fill=none] at (vdots_left) {$\vdots$};
\node[draw=none, font=\large\bfseries, fill=none] at (vdots_mid) {$\vdots$};
\node[draw=none, font=\large\bfseries, fill=none] at (vdots_3) {$\vdots$};
\node[draw=none, font=\large\bfseries, fill=none] at (vdots_right) {$\vdots$};

\draw[thick, decorate, decoration={brace, amplitude=6pt, mirror}]
    (-8.0, 0.3) -- (-8.0, -14.7);
\node[draw=none, fill=none, font=\normalsize\bfseries, anchor=east] at (-8.4,-7.2) {Phase 1};

\draw[thick, decorate, decoration={brace, amplitude=6pt, mirror}]
    (-8.0, -15.3) -- (-8.0, -18.3);
\node[draw=none, fill=none, font=\normalsize\bfseries, anchor=east] at (-8.4,-16.8) {Phase 2};

\draw[thick, decorate, decoration={brace, amplitude=6pt, mirror}]
    (-8.0, -18.6) -- (-8.0, -21.6);
\node[draw=none, fill=none, font=\normalsize\bfseries, anchor=east] at (-8.4,-20.1) {Phase 3};

\end{tikzpicture}
    \caption{Reduction workflow showing the three phases: Phase~1 (variable
    setting through gadgets), Phase~2 (clause selection by Red), and Phase~3
    (literal verification by Blue following dashed paths). Each diamond-shaped
    gadget encodes one QBF variable: the nodes are labeled with letters
    ($a_0, b_0, \ldots, n_0$) for typographic clarity, where the $k$-th gadget
    corresponds to variable $x_k$. For example, the gadget with nodes
    $a_0, a_1, \bar{a}_1, a_2$ encodes variable $x_1$, and the gadget with
    nodes $b_0, b_1, \bar{b}_1, b_2$ encodes variable $x_2$; the branch nodes
    (e.g., $a_1$ and $\bar{a}_1$) represent the true/false assignment choices
    for that variable. The high-level structure mirrors Geography-style
    reductions \cite{schaefer1978, lichtenstein1980}, with the key adaptation
    being the grid-compatible gadget designs.}
    \label{fig:reduction_workflow}
\end{figure}
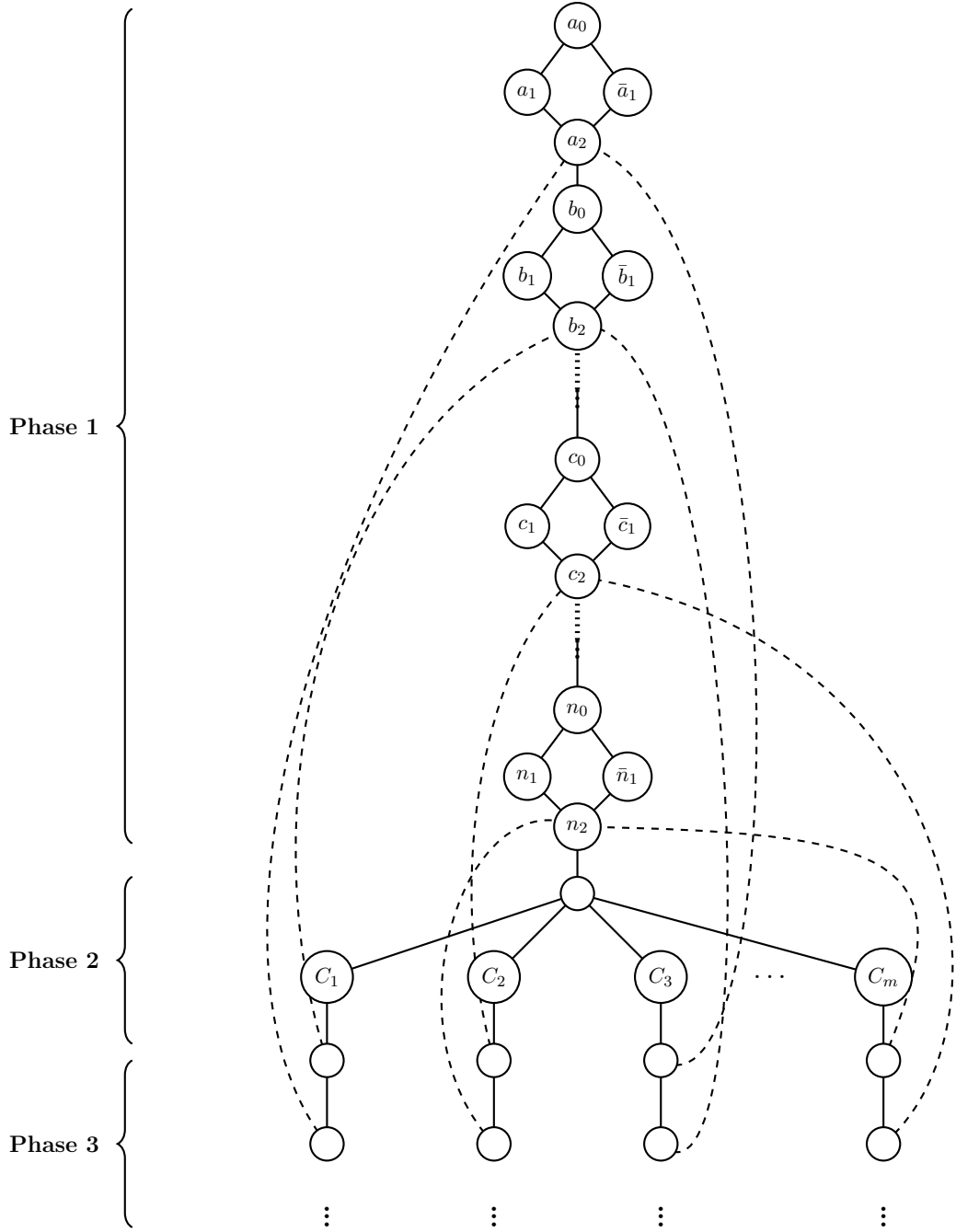

\subsection{Correctness of the Reduction}

\subsubsection{Polynomial Time Construction}

The Cardinal Grid Slime Trail instance is constructible in polynomial time. Each of $n$ variables requires a variable gadget of constant size (See Section~\ref{sec:gadgets}). The choice gadget has size $O(m \cdot n)$ where $m$ is the number of clauses. Each path crossing requires one crossover gadget of constant size, and there are at most $O(m \cdot n)$ crossings. The total graph has $O(m \cdot n)$ vertices and edges, and the construction runs in $O(m \cdot n)$ time.

\subsubsection{Equivalence Proof}

With the construction established, we now prove that it correctly encodes the QBF formula.

\begin{lemma}
Blue has a winning strategy in the constructed Cardinal Grid Slime Trail game if and only if the QBF formula evaluates to true.
\end{lemma}

\begin{proof}
($\Rightarrow$) Suppose Blue has a winning strategy. The variable gadgets force each controlling player to commit to an assignment: Blue for existentially quantified variables, Red for universally quantified ones. Blue's winning strategy through the choice gadget means that for every clause Red selects, Blue can identify a true literal, implying every clause is satisfied under some assignment consistent with the quantifier structure. Therefore, the QBF formula is true.

($\Leftarrow$) Suppose the QBF formula is true. Blue's winning strategy is as follows: in odd variable gadgets, Blue chooses assignments witnessing the existential quantifiers; Red's choices in even variable gadgets are unconstrained but correspond to some universal assignment. Since the formula is true, for every such universal assignment and every clause Red might select, at least one literal in that clause is true. Blue can therefore always follow the dashed path corresponding to a true literal and reach a Blue goal node.
\end{proof}

\section{Grid Gadget Implementations}\label{sec:gadgets}

This section describes the concrete grid constructions that implement each component of the reduction. Following standard complexity-theoretic terminology \cite{garey1979, hearn2009}, a \textit{gadget} is a small subgraph whose local structure enforces a specific logical behavior---such as a player committing to a binary choice, or movement being restricted to one direction---within the larger reduction graph. In each case, we show that the gadget fits on the integer lattice, respects the degree-4 constraint for cardinal movement, maintains correct parity at every decision point, and behaves logically as described in the overview above.

\subsection{Dummy Nodes}

A key element of our grid implementation is the strategic use of \textit{dummy 
nodes}, which are intermediate vertices that maintain proper spacing between functional 
nodes, enable parity adjustment without altering gadget logic, and support 
modular gadget design while respecting grid geometry. In all diagrams throughout 
this section, dummy nodes are represented as dashes along the edges connecting 
functional nodes. Since path length on a 
grid determines whose turn it is at any given node, each gadget must be 
designed so that the intended player holds the turn at every decision point. The gadgets described in the remainder of this 
section are all constructed with this constraint in mind.

\subsection{Odd Variable Gadget}
\label{sec:odd}

The odd variable gadget allows Blue to set odd-indexed variables $(x_1, x_3, x_5, \ldots)$. As shown in Figure~\ref{fig:odd-variable}, the gadget begins at ``Start,'' where it is Blue's turn. Blue may move left to $\alpha_1$ (representing the choice of setting the variable to false) or right to $\beta_1$ (representing the choice of setting it to true). Since both sides are symmetric, we describe only the left path.

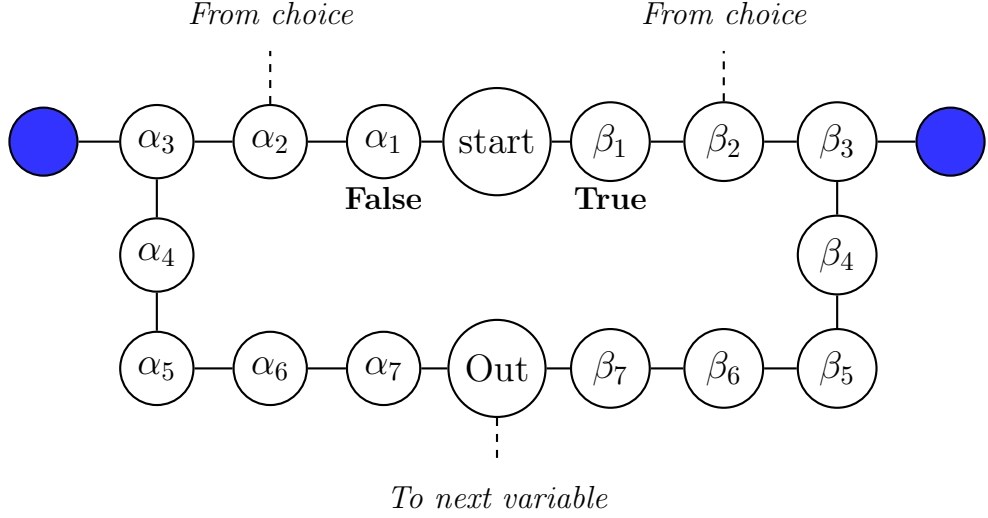
\begin{figure}[htbp]
    \centering
    \begin{tikzpicture}[
        node distance=1.5cm,
        every node/.style={circle, draw, thick, minimum size=0.9cm, font=\large},
        goal/.style={circle, draw, thick, fill=blue!80, minimum size=0.9cm, font=\large},
        connection/.style={draw, thick},
        dashed connection/.style={draw, thick, dashed}
    ]
    
    \node[goal] (bluegoal1) at (0,0) {};
    \node (alpha3) at (1.5,0) {$\alpha_3$};
    \node (alpha2) at (3,0) {$\alpha_2$};
    \node (alpha1) at (4.5,0) {$\alpha_1$};
    \node (start) at (6,0) {start};
    \node (beta1) at (7.5,0) {$\beta_1$};
    \node (beta2) at (9,0) {$\beta_2$};
    \node (beta3) at (10.5,0) {$\beta_3$};
    \node[goal] (bluegoal2) at (12,0) {};
    
    \node (alpha4) at (1.5,-1.5) {$\alpha_4$};
    \node (beta4) at (10.5,-1.5) {$\beta_4$};
    
    \node (alpha5) at (1.5,-3) {$\alpha_5$};
    \node (alpha6) at (3,-3) {$\alpha_6$};
    \node (alpha7) at (4.5,-3) {$\alpha_7$};
    \node (out) at (6,-3) {Out};
    \node (beta7) at (7.5,-3) {$\beta_7$};
    \node (beta6) at (9,-3) {$\beta_6$};
    \node (beta5) at (10.5,-3) {$\beta_5$};
    
    \draw[connection] (bluegoal1) -- (alpha3);
    \draw[connection] (alpha3) -- (alpha2);
    \draw[connection] (alpha2) -- (alpha1);
    \draw[connection] (alpha1) -- (start);
    \draw[connection] (start) -- (beta1);
    \draw[connection] (beta1) -- (beta2);
    \draw[connection] (beta2) -- (beta3);
    \draw[connection] (beta3) -- (bluegoal2);
    
    \draw[connection] (alpha3) -- (alpha4);
    \draw[connection] (beta3) -- (beta4);
    
    \draw[connection] (alpha4) -- (alpha5);
    \draw[connection] (alpha5) -- (alpha6);
    \draw[connection] (alpha6) -- (alpha7);
    \draw[connection] (alpha7) -- (out);
    \draw[connection] (out) -- (beta7);
    \draw[connection] (beta7) -- (beta6);
    \draw[connection] (beta6) -- (beta5);
    \draw[connection] (beta5) -- (beta4);
    
    \draw[dashed connection] (alpha2) -- (3,1.2);
    \draw[dashed connection] (beta2) -- (9,1.2);
    \draw[dashed connection] (out) -- (6,-4.2);
    
    \node[draw=none, font=\normalsize\itshape] at (3,1.7) {From choice};
    \node[draw=none, font=\normalsize\itshape] at (9,1.7) {From choice};
    \node[draw=none, font=\normalsize\bfseries] at (4.5,-0.75) {False};
    \node[draw=none, font=\normalsize\bfseries] at (7.5,-0.75) {True};
    \node[draw=none, font=\normalsize\itshape] at (6,-4.7) {To next variable};
    
    \end{tikzpicture}
    \caption{The odd variable gadget for variables $x_1, x_3, x_5, \ldots, 
    x_{2n+1}$. At Start, it is Blue's turn to move. Blue moves left to set 
    the variable \textit{false} or right to set it \textit{true}.}
    \label{fig:odd-variable}
\end{figure}

After Blue moves to $\alpha_1$ (an odd-distance node from Start, so it is 
Red's turn upon arrival), Red must move to $\alpha_2$. Blue may now 
move to $\alpha_3$ or to the choice gadget. However, moving directly to the 
choice gadget results in an immediate loss for Blue via diode gadget 
(see Section~\ref{sec:diode}) enforces that traversal along verification 
paths flows only from the choice gadget (see Section ~\ref{sec:choice}) toward variable gadgets, not the
reverse. Blue, therefore, moves to $\alpha_3$.

From there, Red can then choose to move to $\alpha_4$ or the Blue target node 
to lose right away. If Red moves to $\alpha_4$, Blue must move to $\alpha_5$, 
Red must then move to $\alpha_6$, Blue must then move to $\alpha_7$ after 
which Red must move to Out. At Out, Blue faces another choice: proceed to the 
next gadget or move right to $\beta_7$. Moving to $\beta_7$ is losing for 
Blue.

Under optimal play by both sides, the game, therefore, proceeds from Out to the next gadget. One side of the gadget will be slimed and inaccessible, while the other remains reachable from the choice gadget.

\subsection{Even Variable Gadget}
\label{sec:even}

The even variable gadget allows Red to set even-indexed variables 
$(x_2, x_4, x_6, \ldots, x_{2n})$. Its structure mirrors that of the odd 
variable gadget, but with Red controlling the critical decision at ``Start'' 
(right to set true, or left to set false). The extra dummy nodes on each side 
of the gadget are necessary to account for the additional turn that elapses 
between gadgets, ensuring that it is Red's turn at ``Start'' rather than Blue's.

As shown in Figure~\ref{fig:even-variable}, the gadget begins at ``Start,'' 
where it is Red's turn. Red may move left to $\alpha_1$ (setting the variable 
to false) or right to $\beta_1$ (setting it to true). We describe the left 
path; the right path is symmetric.

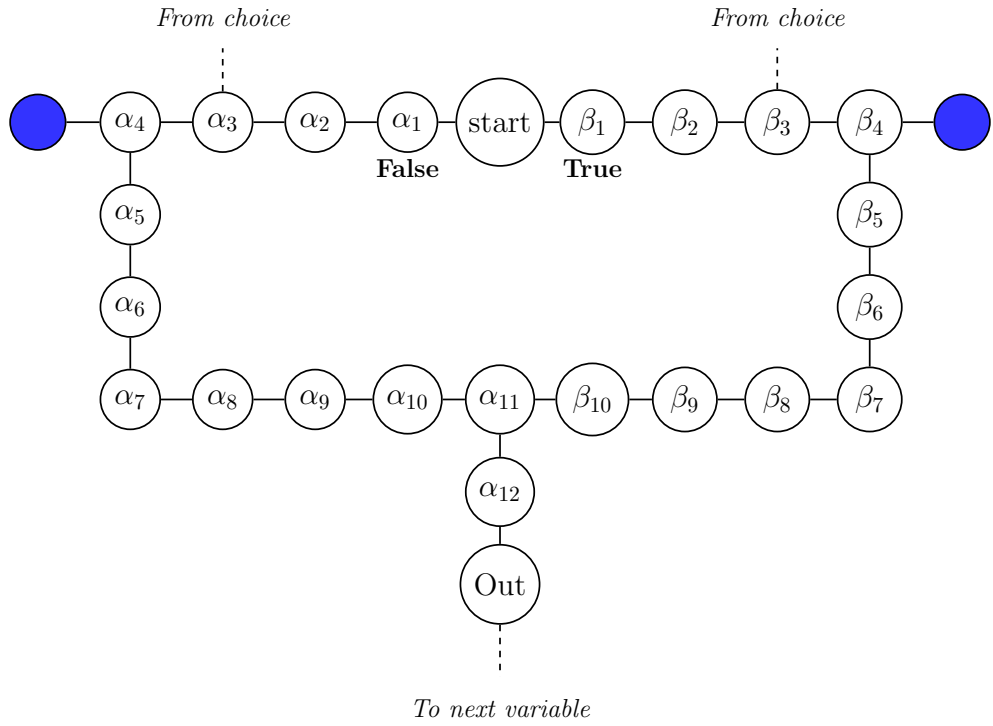
\begin{figure}[htbp]
    \centering
    \resizebox{0.8\textwidth}{!}{
        \begin{tikzpicture}[
        node distance=1.0 cm,
        every node/.style={circle, draw, thick, minimum size=0.5cm, font=\large},
        goal/.style={circle, draw, thick, fill=blue!80, minimum size=0.9cm, font=\large},
        connection/.style={draw, thick},
        dashed connection/.style={draw, thick, dashed}
    ]
    
    \node[goal] (bluegoal1) at (0,0) {};
    \node (alpha4) at (1.5,0) {$\alpha_4$};
    \node (alpha3) at (3,0) {$\alpha_3$};
    \node (alpha2) at (4.5,0) {$\alpha_2$};
    \node (alpha1) at (6,0) {$\alpha_1$};
    \node (start) at (7.5,0) {start};
    \node (beta1) at (9,0) {$\beta_1$};
    \node (beta2) at (10.5,0) {$\beta_2$};
    \node (beta3) at (12,0) {$\beta_3$};
    \node (beta4) at (13.5,0) {$\beta_4$};
    \node[goal] (bluegoal2) at (15,0) {};
    
    \node (alpha5) at (1.5,-1.5) {$\alpha_5$};
    \node (beta5) at (13.5,-1.5) {$\beta_5$};
    
    \node (alpha6) at (1.5,-3) {$\alpha_6$};
    \node (beta6) at (13.5,-3) {$\beta_6$};
    
    \node (alpha7) at (1.5,-4.5) {$\alpha_7$};
    \node (alpha8) at (3,-4.5) {$\alpha_8$};
    \node (alpha9) at (4.5,-4.5) {$\alpha_9$};
    \node (alpha10) at (6,-4.5) {$\alpha_{10}$};
    \node (alpha11) at (7.5,-4.5) {$\alpha_{11}$};
    \node (beta10) at (9,-4.5) {$\beta_{10}$};
    \node (beta9) at (10.5,-4.5) {$\beta_9$};
    \node (beta8) at (12,-4.5) {$\beta_8$};
    \node (beta7) at (13.5,-4.5) {$\beta_7$};
    \node (alpha12) at (7.5,-6.0) {$\alpha_{12}$};
    \node (out) at (7.5,-7.5) {Out};
    
    \draw[connection] (bluegoal1) -- (alpha4);
    \draw[connection] (alpha4) -- (alpha3);
    \draw[connection] (alpha3) -- (alpha2);
    \draw[connection] (alpha2) -- (alpha1);
    \draw[connection] (alpha1) -- (start);
    \draw[connection] (start) -- (beta1);
    \draw[connection] (beta1) -- (beta2);
    \draw[connection] (beta2) -- (beta3);
    \draw[connection] (beta3) -- (beta4);
    \draw[connection] (beta4) -- (bluegoal2);
    
    \draw[connection] (alpha4) -- (alpha5);
    \draw[connection] (alpha5) -- (alpha6);
    \draw[connection] (alpha6) -- (alpha7);
    
    \draw[connection] (beta4) -- (beta5);
    \draw[connection] (beta5) -- (beta6);
    \draw[connection] (beta6) -- (beta7);
    \draw[connection] (alpha11) -- (alpha12);
    \draw[connection] (alpha12) -- (out);
    
    \draw[connection] (alpha7) -- (alpha8);
    \draw[connection] (alpha8) -- (alpha9);
    \draw[connection] (alpha9) -- (alpha10);
    \draw[connection] (alpha10) -- (alpha11);
    \draw[connection] (alpha11) -- (beta10);
    \draw[connection] (beta10) -- (beta9);
    \draw[connection] (beta9) -- (beta8);
    \draw[connection] (beta8) -- (beta7);
    
    \draw[dashed connection] (alpha3) -- (3,1.2);
    \draw[dashed connection] (beta3) -- (12,1.2);
    \draw[dashed connection] (out) -- (7.5,-9.0);
    
    \node[draw=none, font=\normalsize\itshape] at (3,1.7) {From choice};
    \node[draw=none, font=\normalsize\itshape] at (12,1.7) {From choice};
    \node[draw=none, font=\normalsize\bfseries] at (6,-0.75) {False};
    \node[draw=none, font=\normalsize\bfseries] at (9,-0.75) {True};
    \node[draw=none, font=\normalsize\itshape] at (7.5,-9.5) {To next variable};
    
    \end{tikzpicture}
    }
    \caption{The even variable gadget for variables $x_2, x_4, x_6, \ldots, 
    x_{2n}$. The starting position is Start, and it is Red's turn to move. 
    Red moves left to set the variable \textit{false} or right to set it 
    \textit{true}.}
    \label{fig:even-variable}
\end{figure}

After Red moves to $\alpha_1$ (one step from Start, so it is Blue's turn upon 
arrival), Blue must move to $\alpha_2$, and Red must move to $\alpha_3$ (two 
steps from $\alpha_1$, so it is Red's turn). Blue may now move to $\alpha_4$ 
or to the choice gadget. However, moving directly to the choice gadget results 
in an immediate loss for Blue: the diode gadget (described in 
Section~\ref{sec:diode}) enforces that traversal along verification 
paths flows only from the choice gadget toward variable gadgets, not the 
reverse. Attempting to enter the choice gadget from a variable gadget therefore 
traps Blue with no legal path forward. Blue therefore moves to $\alpha_4$. Red 
may then choose to move to $\alpha_5$ or to Blue's target node and lose right 
away.

If Red moves to $\alpha_5$, Blue must move to $\alpha_6$ (Blue's turn, one 
step from $\alpha_5$), Red moves to $\alpha_7$, Blue moves to $\alpha_8$, Red 
moves to $\alpha_9$, Blue moves to $\alpha_{10}$, and Red must then move to 
$\alpha_{11}$ --- each alternation following directly from the parity 
established at Start. From here, Blue moves to $\alpha_{12}$ to avoid being 
trapped, and Red must then move to Out. From Out, Blue proceeds to the next 
gadget.

As with the odd variable gadget, optimal play results in one side being slimed 
and the other remaining accessible for later verification.

\subsection{Choice Gadget}
\label{sec:choice}

\subsubsection{The Choice Gadget}

Figure~\ref{fig:choice-gadget} illustrates the choice gadget. The choice gadget serves as Red's adversarial testing mechanism. After all variables have been set in Phase 1, Red enters the choice gadget and strategically selects a single clause from the formula to force Blue to satisfy. This selection is entirely Red's choice, and Blue has no influence over which clause Red picks.


\begin{figure}[htbp]
    \centering
    \resizebox{1.05\textwidth}{!}{\input{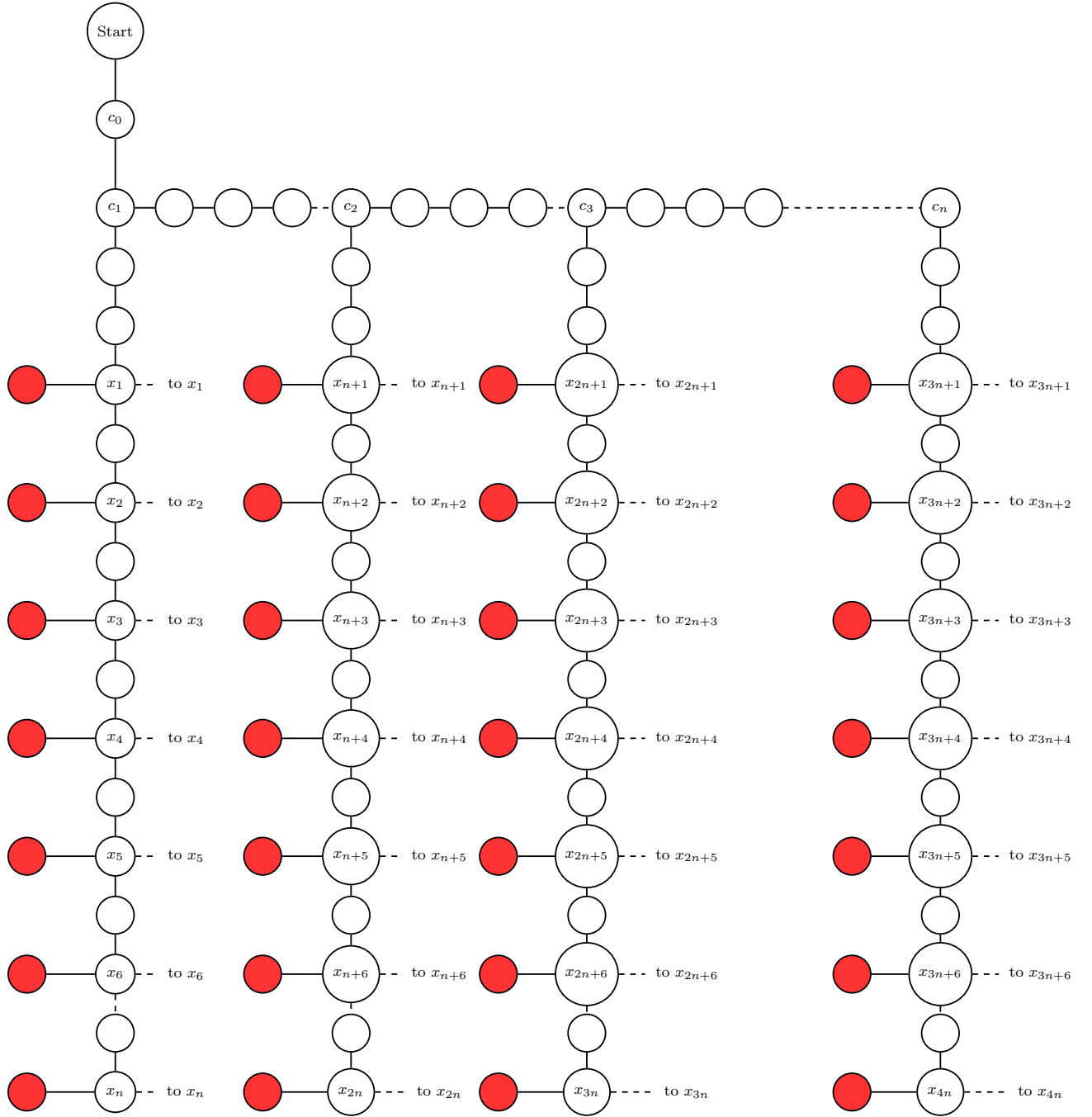}}
     \caption{The choice gadget. The starting position is Start, and it is Red's turn to move.}
    \label{fig:choice-gadget}
\end{figure}

The gadget operates as follows: Red enters at ``Start'' and moves to $c_0$, then Blue moves to $c_1$. From $c_1$, Red can either move directly down to test the first clause, or move horizontally through intermediate nodes to reach any later clause. Once Red selects a clause $C_i$, Red is positioned at the first literal node $x_i$, and Blue must decide whether to follow the dashed connection back to the corresponding variable gadget to verify the literal's truth value, skip to the next literal in the clause, or concede by moving to Red's goal node.

The dashed connections from the clause back to the variable gadgets are crucial. Each dashed connection corresponds to a literal in the selected clause. During Phase 1, Blue set the variables and in doing so slimed one side of each variable gadget. When Blue follows a dashed connection, one of two outcomes occurs.

If the connection leads to an unslimed path in the variable gadget, the literal evaluates to true. Blue can traverse this path through a diode gadget (see Section~\ref{sec:diode}) and reach a Blue goal node, winning immediately.

If the connection leads to a slimed path, the literal evaluates to false. The path is blocked, and Blue must go to the Red target node from the diode gadget.

Red's strategy is to select a clause $C_i$ for which Red can determine that every literal evaluates to false given the variable assignments from Phase 1. Since the game has perfect information, Red can inspect the slimed paths and make this determination before committing to a clause. If Red can find such a clause, Blue has no winning move and must enter Red's goal node, losing the game.

Conversely, if Blue can find at least one true literal, Blue wins. Blue has a winning strategy in the overall game if and only if Blue can set the existential variables such that no matter which clause Red selects, at least one literal in that clause is true---precisely the condition for the QBF formula to evaluate to true.

\subsubsection{The Simplified Choice Gadget: Mathematical Structure}

To illustrate the core mechanics more clearly, we examine a simplified instance shown in Figure~\ref{fig:simplified-choice}. Consider the simplified gadget with two clauses $C_1$ and $C_2$, each containing two literals. Let the boolean formula be
\begin{equation}
\phi = C_1 \wedge C_2
\end{equation}
where
\begin{align}
C_1 &= (x_{11} \vee x_{12}) \\
C_2 &= (x_{21} \vee x_{22})
\end{align}
and $x_{11}, x_{12}, x_{21}, x_{22}$ are the four literals, where $x_{ij}$ denotes the $j$-th literal of clause $C_i$. Here $x_{11}$ and $x_{21}$ are existentially quantified (controlled by Blue) and $x_{12}$ and $x_{22}$ are universally quantified (controlled by Red).

\subsubsection{Red's Clause Selection and Its Adversarial Structure}

Red enters the gadget at ``Start'' and moves to $c_0$. Blue then moves to $c_1$. At this point, Red possesses complete freedom in clause selection. Since the game has perfect information, Red can determine which clause will be hardest for Blue to satisfy before committing to a choice. Red's decision is expressed formally as the choice of index $i \in \{1, 2\}$: Red can either move directly downward to test $C_1$ (if Red chooses $i = 1$), or move rightward through intermediate nodes $c_{11}, c_{12}, c_{13}$ and then downward to reach $C_2$ (if Red chooses $i = 2$).

Suppose Red selects $C_2$ by moving right and then positioning the token at the first literal node $x_{21}$ within that clause.

\subsubsection{Blue's Literal Verification and Truth Evaluation}

Blue now faces the necessity of responding to Red's clause selection. Blue examines the current literal $x_{21}$ in clause $C_2$ and must choose one of three actions:

(1) Follow the dashed connection back to the variable gadget corresponding to $x_{21}$ to verify its truth value.

(2) Skip to the next literal $x_{22}$ within the same clause.

(3) Concede by moving to Red's goal node, accepting immediate defeat.

When Blue follows the dashed connection for $x_{21}$, the outcome is determined entirely by the variable assignment phase. During Phase 1, when Blue set the variables, Blue simultaneously determined which paths in each variable gadget would become slimed (inaccessible) and which would remain unslimed (accessible).

Formally, for a literal $x_{ij}$, define the evaluation function:
\begin{equation}
\text{eval}(x_{ij}) = \begin{cases}
\text{true} & \text{if the path corresponding to } x_{ij} \text{ is unslimed} \\
\text{false} & \text{if the path corresponding to } x_{ij} \text{ is slimed}
\end{cases}
\end{equation}

If $\text{eval}(x_{21}) = \text{true}$, then the dashed connection leads to an unslimed section of the variable gadget. Blue can traverse this accessible path through a diode gadget and ultimately reach a Blue goal node, immediately winning the game.

If $\text{eval}(x_{21}) = \text{false}$, then the dashed connection leads to a slimed section of the variable gadget. The path is blocked and impassable. Blue cannot proceed further along this route and must return to the clause to attempt the next literal.

\subsubsection{Clause Satisfaction and Red's Winning Condition}

Consider the full clause $C_2 = (x_{21} \vee x_{22})$. For Blue to satisfy this clause, Blue must find at least one literal whose evaluation is true. Mathematically, this condition is expressed as:
\begin{equation}
\text{eval}(x_{21}) \vee \text{eval}(x_{22}) = \text{true}
\end{equation}

If both $\text{eval}(x_{21}) = \text{false}$ and $\text{eval}(x_{22}) = \text{false}$, then both dashed connections are slimed. Blue's only remaining option is to move to Red's goal node, resulting in an immediate loss for Blue and immediate victory for Red.

Red's adversarial strategy is therefore to select a clause in which all literals evaluate to false given the current variable assignments. In our example, Red would choose $C_2$ only if both $\text{eval}(x_{21}) = \text{false}$ and $\text{eval}(x_{22}) = \text{false}$. More generally, Red seeks a clause $C_i$ such that every literal in that clause evaluates to false. If Red can successfully force Blue into such a clause, Blue cannot satisfy it and Red wins.

\subsubsection{Blue's Winning Strategy and QBF Equivalence}

Blue's overall winning strategy in the game requires that Blue can set the existential variables during Phase 1 such that the following condition holds: for every clause $C_i$ that Red might select, at least one literal in $C_i$ evaluates to true. Formally, Blue wins if and only if Blue can assign values to the existentially quantified variables such that:
\begin{equation}
\forall\, i \in \{1, 2, \ldots, m\} : \bigvee_{j=1}^{k_i} \text{eval}(x_{ij}) = \text{true}
\end{equation}
where $m$ is the total number of clauses in the formula.

This condition is precisely equivalent to the statement that the boolean formula $\phi$ evaluates to true when all existentially quantified variables are set according to Blue's Phase 1 choices and all universally quantified variables are set according to Red's Phase 1 choices.

Therefore, Blue possesses a winning strategy in the constructed Cardinal Grid Slime Trail game if and only if the QBF formula $\exists\, x_1 \; \forall\, x_2 \; \cdots \; Q_n\, x_n \; : \; \phi(x_1, x_2, \ldots, x_n)$ evaluates to true. This establishes the reduction from QBF to Cardinal Grid Slime Trail and completes the PSPACE-hardness argument.

\begin{figure}[htbp]
    \centering
    \begin{tikzpicture}[
    node distance=1.3cm,
    every node/.style={circle, draw, thick, minimum size=0.7cm, font=\scriptsize},
    goal/.style={circle, draw, thick, fill=blue!80, minimum size=0.7cm, font=\scriptsize},
    redgoal/.style={circle, draw, thick, fill=red!80, minimum size=0.7cm, font=\scriptsize},
    connection/.style={draw, thick},
    dashed connection/.style={draw, thick, dashed}
]
\def\s{1.1}
\node (start) at (0, 6*\s) {Start};
\node (c0) at (0, 4.5*\s) {$c_0$};
\node (c1) at (0, 3*\s) {$c_1$};
\node (c11) at (\s, 3*\s) {$c_{11}$};
\node (c12) at (2*\s, 3*\s) {$c_{12}$};
\node (c13) at (3*\s, 3*\s) {$c_{13}$};
\node (c2) at (4*\s, 3*\s) {$c_2$};
\node (t11) at (0, 2*\s) {$t_{11}$};
\node (t12) at (0, \s) {$t_{12}$};
\node (x11) at (0, 0) {$x_{11}$};
\node[redgoal] (red_x11) at (-1.5*\s, 0) {};
\node (t13) at (0, -\s) {$t_{13}$};
\node (x12) at (0, -2*\s) {$x_{12}$};
\node[redgoal] (red_x12) at (-1.5*\s, -2*\s) {};
\node (t21) at (4*\s, 2*\s) {$t_{21}$};
\node (t22) at (4*\s, \s) {$t_{22}$};
\node (x21) at (4*\s, 0) {$x_{21}$};
\node[redgoal] (red_x21) at (2.5*\s, 0) {};
\node (t23) at (4*\s, -\s) {$t_{23}$};
\node (x22) at (4*\s, -2*\s) {$x_{22}$};
\node[redgoal] (red_x22) at (2.5*\s, -2*\s) {};
\draw[connection] (start) -- (c0);
\draw[connection] (c0) -- (c1);
\draw[connection] (c1) -- (c11);
\draw[connection] (c11) -- (c12);
\draw[connection] (c12) -- (c13);
\draw[connection] (c13) -- (c2);
\draw[connection] (c1) -- (t11);
\draw[connection] (t11) -- (t12);
\draw[connection] (t12) -- (x11);
\draw[connection] (x11) -- (t13);
\draw[connection] (t13) -- (x12);
\draw[connection] (c2) -- (t21);
\draw[connection] (t21) -- (t22);
\draw[connection] (t22) -- (x21);
\draw[connection] (x21) -- (t23);
\draw[connection] (t23) -- (x22);
\draw[connection] (red_x11) -- (x11);
\draw[connection] (red_x12) -- (x12);
\draw[connection] (red_x21) -- (x21);
\draw[connection] (red_x22) -- (x22);
\node[draw=none, font=\scriptsize] (to_x11) at (1.5*\s, 0) {to $x_{11}$};
\draw[dashed connection] (x11) -- (to_x11);
\node[draw=none, font=\scriptsize] (to_x12) at (1.5*\s, -2*\s) {to $x_{12}$};
\draw[dashed connection] (x12) -- (to_x12);
\node[draw=none, font=\scriptsize] (to_x21) at (5.5*\s, 0) {to $x_{21}$};
\draw[dashed connection] (x21) -- (to_x21);
\node[draw=none, font=\scriptsize] (to_x22) at (5.5*\s, -2*\s) {to $x_{22}$};
\draw[dashed connection] (x22) -- (to_x22);
\end{tikzpicture}
     \caption{A simplified choice gadget showing the clause selection mechanism with two clauses $C_1 = (x_{11} \vee x_{12})$ and $C_2 = (x_{21} \vee x_{22})$, each containing two literals.}
    \label{fig:simplified-choice}
\end{figure}
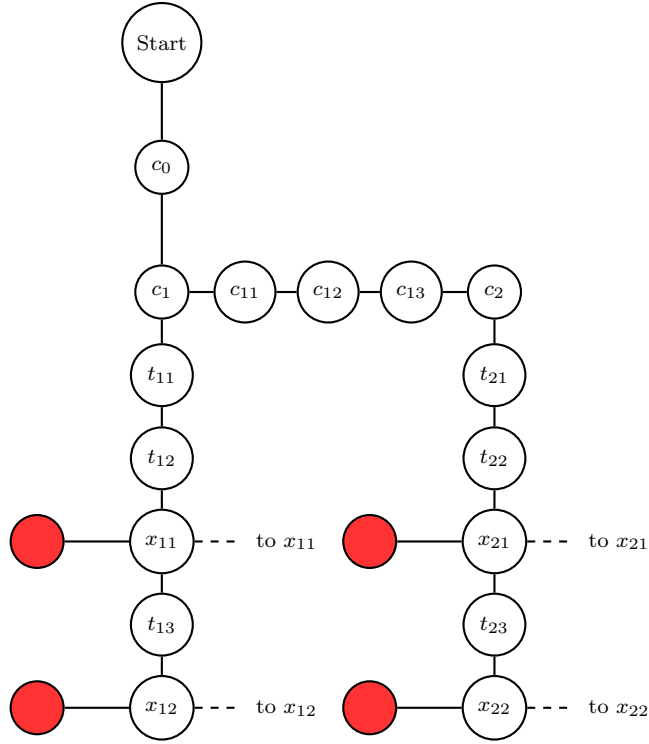

\subsection{Merge Gadget}

The choice gadget described above works straightforwardly when each variable 
appears at most once in the formula. However, a variable may appear in multiple 
clauses, requiring multiple connections from distinct clauses back to the same 
variable gadget. Since each vertex in a cardinal grid graph has degree at most 4, 
attaching multiple incoming paths requires care. The 
merge gadget resolves the fan-in problem by providing a structured way to 
combine these multiple incoming paths.

As shown in Figure~\ref{fig:side-by-side}, the merge gadget has the structure 
of a unary tree that grows with the number of times a variable appears. 
Figure~\ref{fig:merge-gadget} gives a concrete example for a variable appearing 
in three clauses: three separate paths lead to $\zeta_{10}$, $\zeta_{20}$, and 
$\zeta_{30}$, and all three converge at $\mu_4$, which then connects to the 
variable gadget via a single edge.

A critical requirement of the merge gadget is that each incoming path must 
only be traversable in one direction: from the choice gadget toward the 
variable gadget and never the reverse. To enforce this, a diode gadget 
(described in Section~\ref{sec:diode}) is placed at each tip of the 
merge gadget, that is, at each of the entry points$\zeta_{10}$, $\zeta_{20}$, and 
$\zeta_{30}$, 
and $\eta_{30}$. This ensures that all paths through the merge gadget flow 
strictly from the clauses toward the variable gadget, preserving the integrity 
of the reduction.

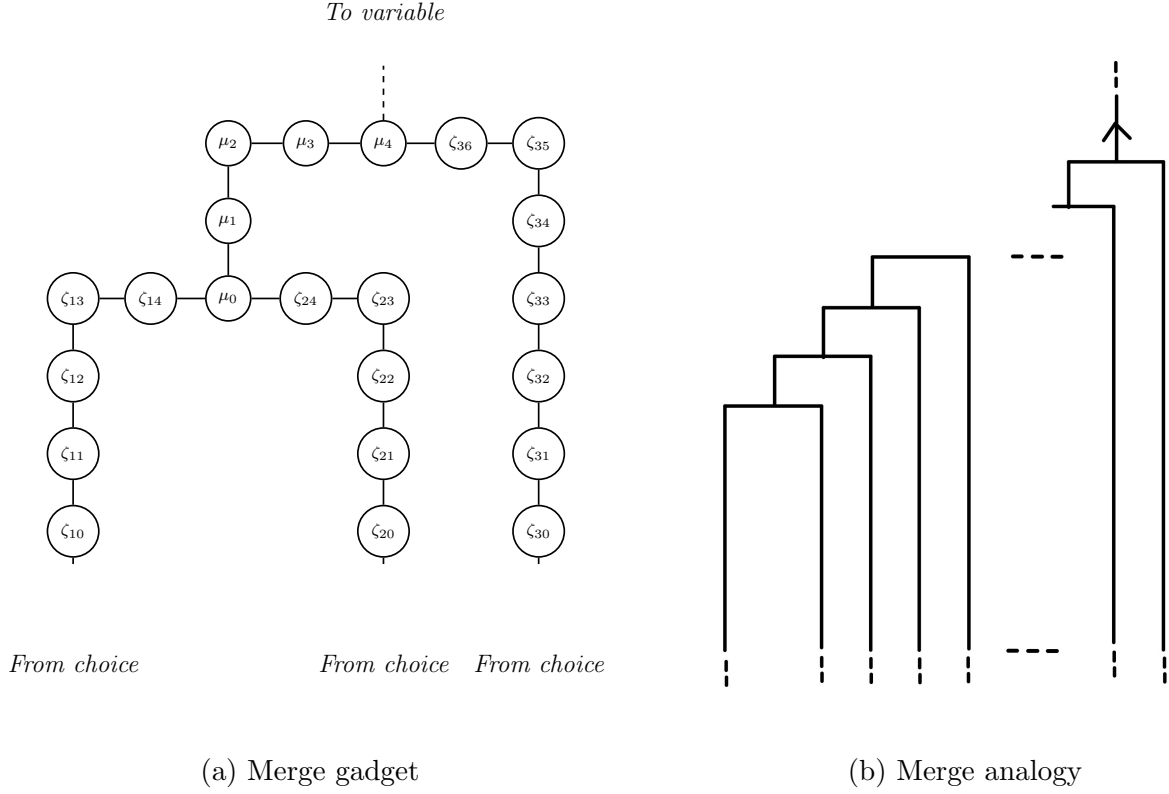
\begin{figure}[h!]
    \centering
    \begin{subfigure}[b]{0.5\textwidth}
    \centering
        \resizebox{1\textwidth}{!}{\begin{tikzpicture}[
    node distance=1.3cm,
    every node/.style={circle, draw, thick, minimum size=0.7cm, font=\scriptsize},
    connection/.style={draw, thick},
    dashed connection/.style={draw, thick, dashed}
]
\def\s{1.3}

\node (l2) at (0, -2*\s) {$\zeta_{23}$};
\node (l3) at (-\s, -2*\s) {$\zeta_{24}$};
\node (l4) at (-2*\s, -2*\s) {$\mu_0$};
\node (l5) at (-3*\s, -2*\s) {$\zeta_{14}$};
\node (l6) at (-4*\s, -2*\s) {$\zeta_{13}$};
\node (l7) at (-4*\s, -3*\s) {$\zeta_{12}$};
\node (l8) at (-4*\s, -4*\s) {$\zeta_{11}$};
\node (l9) at (-4*\s, -5*\s) {$\zeta_{10}$};

\node[draw=none, font=\normalsize\itshape, align=center] at (-4*\s, -6.7*\s) {From choice};
\draw[dashed connection] (l9) -- (-4*\s, -5.5*\s);

\node (d1) at (-2*\s, -\s) {$\mu_1$};
\node (d2) at (-2*\s, 0) {$\mu_2$};

\node (e1) at (-\s, 0) {$\mu_3$};
\node (e2) at (0, 0) {$\mu_4$};
\node (e3) at (\s, 0) {$\zeta_{36}$};
\node (e4) at (2*\s, 0) {$\zeta_{35}$};

\node[draw=none, font=\normalsize\itshape, align=center] at (0, 1.7*\s) {To variable};
\draw[dashed connection] (e2) -- (0, \s);

\node (z1) at (2*\s, -\s) {$\zeta_{34}$};
\node (z2) at (2*\s, -2*\s) {$\zeta_{33}$};
\node (z3) at (2*\s, -3*\s) {$\zeta_{32}$};
\node (z4) at (2*\s, -4*\s) {$\zeta_{31}$};
\node (z5) at (2*\s, -5*\s) {$\zeta_{30}$};

\node[draw=none, font=\normalsize\itshape, align=center] at (2*\s, -6.7*\s) {From choice};
\draw[dashed connection] (z5) -- (2*\s, -5.5*\s);

\node (g1) at (0, -3*\s) {$\zeta_{22}$};
\node (g2) at (0, -4*\s) {$\zeta_{21}$};
\node (g3) at (0, -5*\s) {$\zeta_{20}$};

\node[draw=none, font=\normalsize\itshape, align=center] at (0, -6.7*\s) {From choice};
\draw[dashed connection] (g3) -- (0, -5.5*\s);


\draw[connection] (l2) -- (l3) -- (l4) -- (l5) -- (l6) -- (l7) -- (l8) -- (l9);

\draw[connection] (l2) -- (g1) -- (g2) -- (g3);

\draw[connection] (d2) -- (d1) -- (l4);

\draw[connection] (d2) -- (e1) -- (e2) -- (e3) -- (e4);

\draw[connection] (e4) -- (z1) -- (z2) -- (z3) -- (z4) -- (z5);

\end{tikzpicture}}
        \caption{Merge gadget}
        \label{fig:merge-gadget}
    \end{subfigure}
    \hfill
    \begin{subfigure}[b]{0.45\textwidth}
        \centering
        \includegraphics[width=\textwidth]{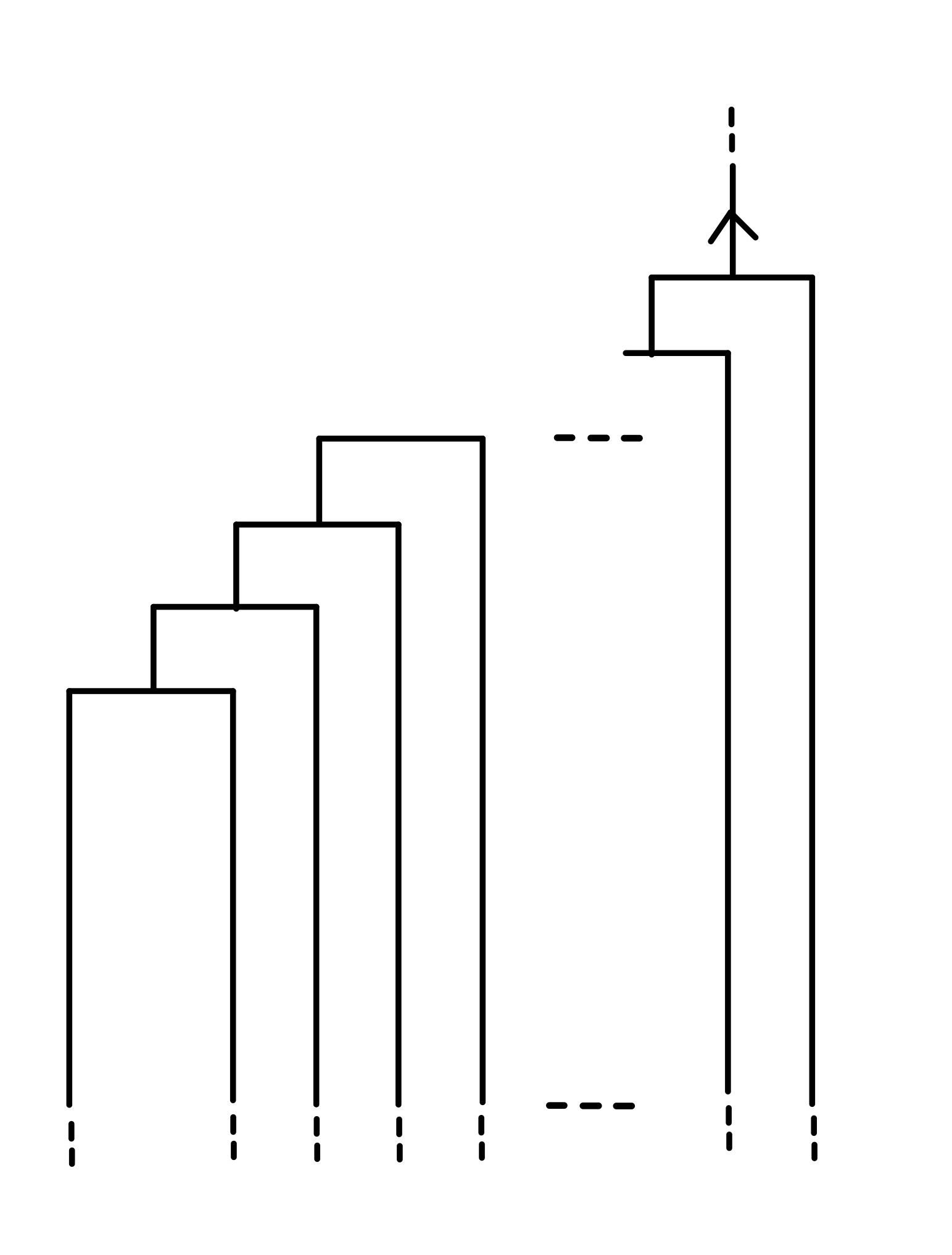}
        \caption{Merge analogy}
        \label{fig:merge-analogy}
    \end{subfigure}
    \caption{Merge gadget (a) and its unary tree analogy (b).}
    \label{fig:side-by-side}
\end{figure}

\subsection{Diode Gadget}
\label{sec:diode}

The diode gadget, shown in Figure~\ref{fig:diode-gadget}, 
enforces that movement along any directed connection flows in only one intended 
direction. The name comes by analogy with the electrical diode, a component 
that permits current to flow in only one direction \cite{sedra2004}. It is used 
in multiple places throughout the reduction---between 
clauses and variable gadgets, and at the tips of merge gadgets---to prevent 
players from exploiting connections to travel backward and reach unintended 
parts of the graph.

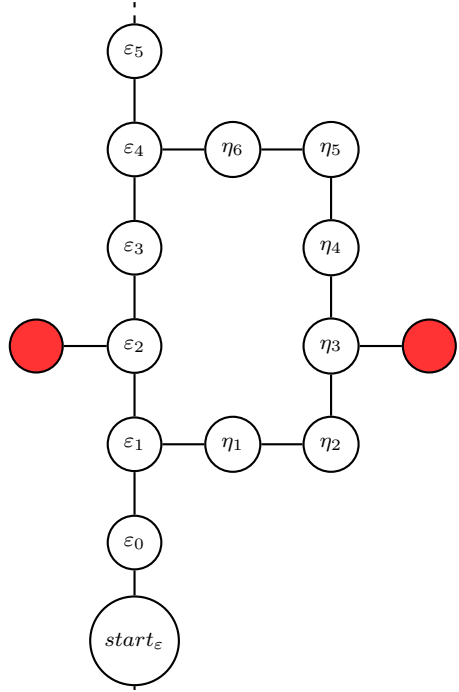
\begin{figure}[htbp]
    \centering
\begin{tikzpicture}[
    node distance=1.3cm,
    every node/.style={circle, draw, thick, minimum size=0.7cm, font=\scriptsize},
    goal/.style={circle, draw, thick, fill=blue!80, minimum size=0.7cm, font=\scriptsize},
    redgoal/.style={circle, draw, thick, fill=red!80, minimum size=0.7cm, font=\scriptsize},
    connection/.style={draw, thick},
    dashed connection/.style={draw, thick, dashed}
]
\def\s{1.3}
\node (epsilon0) at (7*\s, -17*\s) {$\varepsilon_0$};
\node (startepsilon) at (7*\s, -18*\s) {$start_\varepsilon$};
\draw[dashed connection] (startepsilon) -- (7*\s, -18.5*\s);

\node (epsilon1) at (7*\s, -16*\s) {$\varepsilon_1$};
\node (epsilon2) at (7*\s, -15*\s) {$\varepsilon_2$};
\node (epsilon3) at (7*\s, -14*\s) {$\varepsilon_3$};
\node (epsilon4) at (7*\s, -13*\s) {$\varepsilon_4$};
\node (epsilon5) at (7*\s, -12*\s) {$\varepsilon_5$};

\draw[dashed connection] (epsilon5) -- (7*\s, -11.5*\s);
\node (eta1) at (8*\s, -16*\s) {$\eta_1$};
\node (eta2) at (9*\s, -16*\s) {$\eta_2$};
\node (eta3) at (9*\s, -15*\s) {$\eta_3$};
\node (eta4) at (9*\s, -14*\s) {$\eta_4$};
\node (eta5) at (9*\s, -13*\s) {$\eta_5$};
\node (eta6) at (8*\s, -13*\s) {$\eta_6$};
\node[redgoal] (red1) at (6*\s, -15*\s) {};
\node[redgoal] (red2) at (10*\s, -15*\s) {};
\draw[connection] (startepsilon) -- (epsilon0) -- (epsilon1) -- (epsilon2) -- (epsilon3) -- (epsilon4) -- (epsilon5);
\draw[connection] (epsilon1) -- (eta1) -- (eta2);
\draw[connection] (eta2) -- (eta3) -- (eta4) -- (eta5);
\draw[connection] (eta5) -- (eta6) -- (epsilon4);
\draw[connection] (epsilon2) -- (red1);
\draw[connection] (eta3) -- (red2);
\end{tikzpicture}
    \caption{The diode gadget ensures that traversal flows in only one 
    direction, preventing players from exploiting connections to reach unintended 
    parts of the graph.}
    \label{fig:diode-gadget}
\end{figure}

The correct path runs from $start_{\varepsilon}$ to $\varepsilon_5$. Blue starts at $start_{\varepsilon}$, 
Red moves to $\varepsilon_0$, Blue moves to $\varepsilon_1$, and optimal play 
then proceeds $\varepsilon_2 \to \varepsilon_3 \to \varepsilon_4 \to 
\varepsilon_5 \to$ exit.

If Red deviates at $\varepsilon_1$ by moving to $\eta_1$, Blue moves to 
$\eta_2$, Red to $\eta_3$, and Blue either concedes immediately or moves to 
$\eta_4$. Continuing along this detour leads back to $\varepsilon_4$, from 
which Blue moves to $\varepsilon_5$ and Red proceeds to the exit as intended. 
The deviation thus accomplishes nothing.

Consider the case when Blue attempts the reverse journey, starting from the 
exit and trying to reach the entry. The gadget forces Blue into a losing 
position. From the exit, Blue moves to $\varepsilon_5$, Red to $\varepsilon_4$. 
If Blue moves to $\eta_6$, the detour through $\eta_5 \to \eta_4 \to \eta_3 
\to \eta_2 \to \eta_1 \to \varepsilon_1 \to \varepsilon_2$ leads to Blue 
eventually being forced onto Red's goal node and losing. Similarly, if Blue 
moves to $\varepsilon_3$ from $\varepsilon_4$, Red responds with $\varepsilon_2$, 
then Blue moves to $\varepsilon_1$, and Red creates a trap by going to $\eta_1$, 
again forcing Blue onto Red's goal node. In every case, attempting to traverse 
the gadget in the reverse direction leads to a loss for Blue.

\subsection{Crossover Gadget}

As we route dashed connections from clauses back to variable gadgets, 
some of these paths will inevitably cross one another on the grid. A direct 
intersection would allow players to make unintended moves, potentially winning 
through routes that have no logical meaning in the reduction. The crossover 
gadget, shown in Figure~\ref{fig:crossover}, resolves these intersections while 
preserving planarity, correct turn order, and the integrity of both crossing 
paths. Figure~\ref{fig:crossover} is the complete and sufficient description of 
the gadget itself. The colored dots in Figure~\ref{fig:crossoverdots} are provided solely as a 
visual aid: they are not part of the 
gadget but are added to make it easier to track turn order and follow the case 
analysis in the sections that follow.

\begin{figure}[p]
    \centering
    \resizebox{1.0\textwidth}{!}{\input{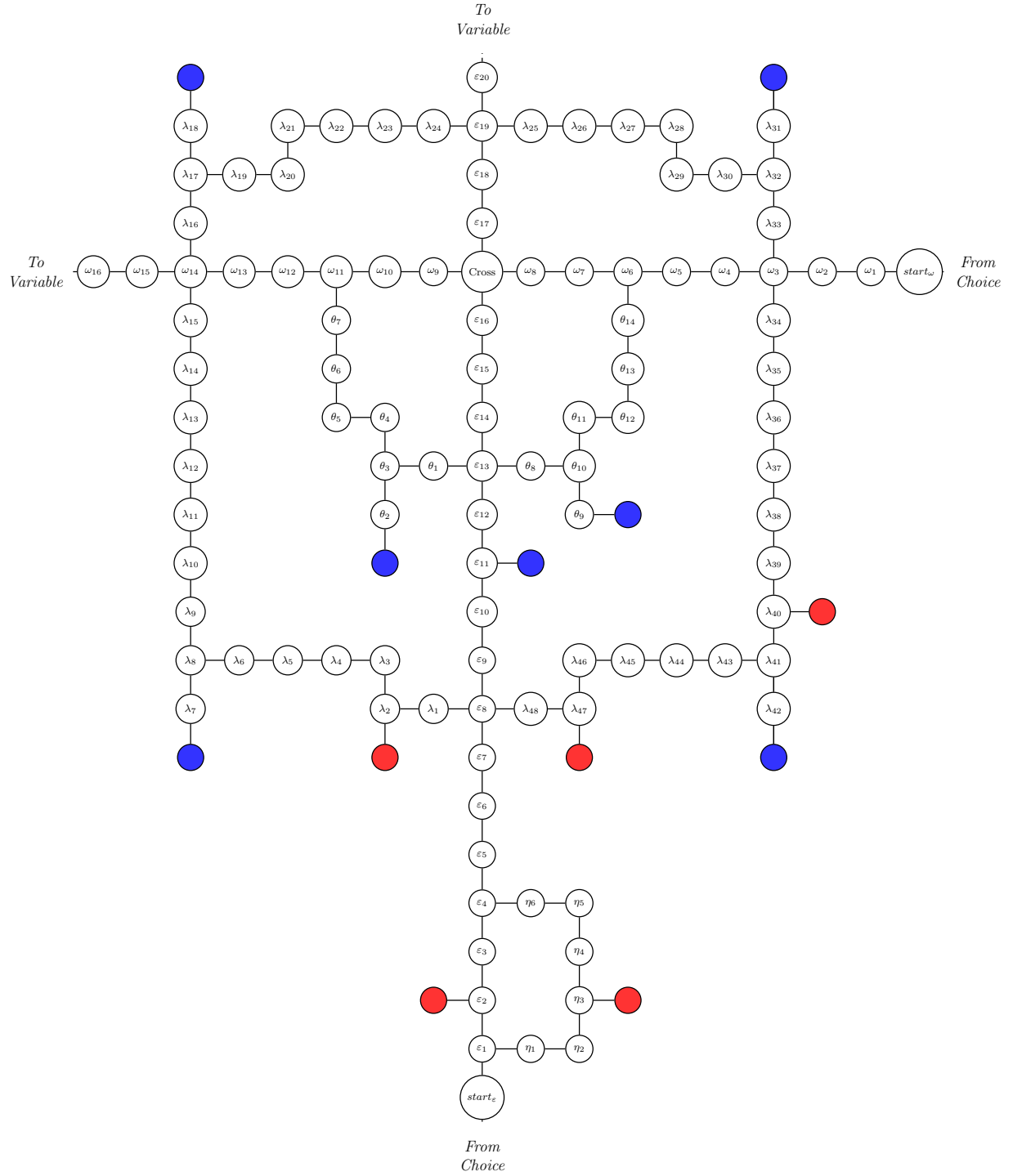}}
    \caption{The crossover gadget. Two entry points ($\mathit{start}_{\varepsilon
    }$ and $\mathit{start}_{\omega}$) 
    and two exit points allow paths to cross without actual intersection while 
    maintaining correct parity.}
    \label{fig:crossover}
\end{figure}

\begin{figure}[p]
    \centering
    \resizebox{1.0\textwidth}{!}{\input{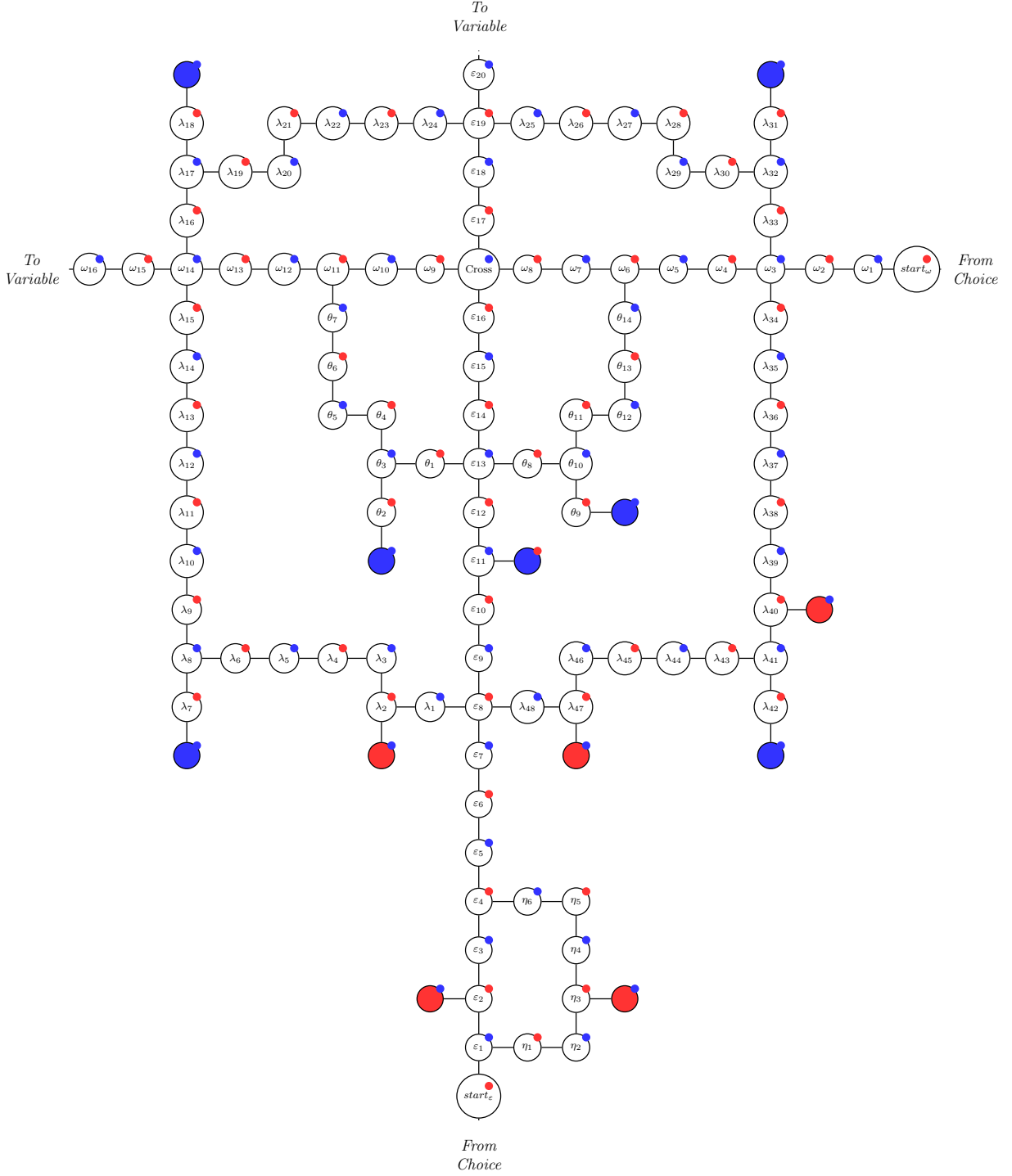}}
    \caption{The crossover gadget with colored dots indicating turn order. 
    Every node carries a small dot in its upper-right corner: a blue dot means 
    it is Red's turn to move away from that node, and a red dot means it is 
    Blue's turn to move away from it. This figure serves as a reference for 
    following the case analysis.}
    \label{fig:crossoverdots}
\end{figure}

A legal traversal of the crossover gadget goes straight through: a path 
entering from $\mathit{start}_{\varepsilon}$ exits at the far end of the same logical line, and a 
path entering from $\mathit{start}_{\omega}$ exits at the far end of the other. The two paths 
share a central ``Cross'' node but are designed so that any deviation from the 
straight-through route leads to a loss for the deviating player.

\textbf{Notation for colored dots.} As shown in Figure~\ref{fig:crossoverdots}, 
every node in the crossover gadget carries a small colored dot in its 
upper-right corner. These dots use a proper 2-coloring of the graph, meaning 
every pair of adjacent nodes receives opposite-colored dots. A blue dot means 
it is Red's turn to move away from that node, and a red dot means it is Blue's 
turn to move away from it. This coloring makes it straightforward to trace 
whose turn it is at any point in the analysis without manually counting steps 
from the start. The colored dots are preserved in all subsequent case analysis 
figures.

\textbf{Notation for shaded paths.} In addition to the colored dots, the case 
analysis figures use shaded paths to illustrate the possible responses available 
to the opposing player after an illegal move is made. The shaded paths and 
their colors are purely for explanation purposes and are not part of the gadget 
itself. Each distinct shaded color represents a different response strategy the 
opposing player may take. In cases where multiple colored paths are shown, each 
path represents an independent scenario branching from the point where the 
opposing player makes a decision. All shaded paths lead to the same outcome: 
the deviating player loses, regardless of which response the opposing player 
chooses.

There are 18 possible illegal moves in total, initiated by either Red or Blue. 
For each illegal move, the opposing player has at least one response that either 
wins the game outright or steers the deviating player back onto the correct 
path. The complete case analysis demonstrated by 18 figures with dots and shaded paths is provided in the Appendix. We present one 
representative case here, which is shown in Figure~\ref{fig:case-1}.

\subsubsection{Example: Blue Deviating at $\varepsilon_8$ (Case 1)}

\begin{figure}[p]
    \centering
    \includegraphics[width=1\textwidth]{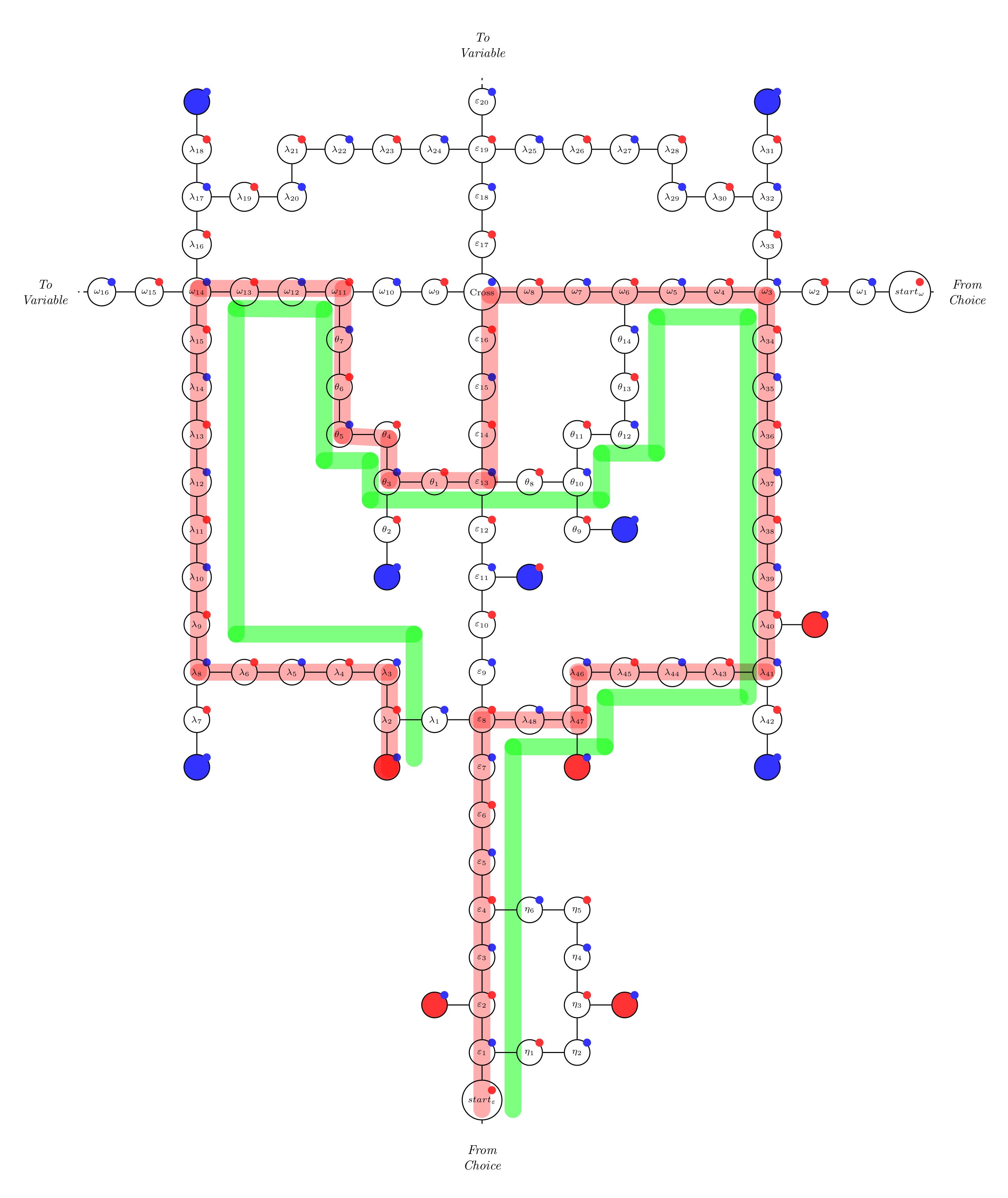}
    \caption{Case 1: Blue deviating at $\varepsilon_8$ by going to the right. The green and red shaded paths show two possible responses Red can take; both result in Blue losing.}
    \label{fig:case-1}
\end{figure}

We now examine the scenario in which Blue deviates at $\varepsilon_8$. Under normal play, $\varepsilon_8$ carries a Red dot, indicating it is Blue's turn to move \emph{away} from this node. 

The scenario is as follows. Red starts at $\mathit{start}_{\varepsilon}$ from the choice gadget, and it is Blue's turn to move. Blue moves $\varepsilon_1 \to \varepsilon_2 \to \varepsilon_3 \to \cdots$ until Red moves to $\varepsilon_8$ and it is Blue's turn. Instead of the legal move to $\varepsilon_9$, Blue moves to $\lambda_{48}$, which is an illegal deviation. Note that $\lambda_{48}$ carries a blue dot, confirming it is a node where Red moves away.

Red must now respond. Red moves to $\lambda_{47}$. Blue, not wishing to lose immediately, avoids the Red goal node below and moves to $\lambda_{46}$. The players alternate along the bottom horizontal path until Blue moves to $\lambda_{41}$. Red then moves up to $\lambda_{40}$ (which carries a red dot, so it is Blue's turn there), and Blue avoids the Red goal node by moving to $\lambda_{39}$. The players continue along the right-side vertical path to $\omega_3$, where it is Red's turn to move. Red moves left to $\omega_4$, Blue to $\omega_5$, Red to $\omega_6$. At this point the path diverges into two branches shown in Figure~\ref{fig:case-1}.

Following the red path: from $\omega_6$, Blue moves to $\omega_7$, Red to $\omega_8$, and Blue to Cross. Red descends to $\varepsilon_{16}$ and the players alternate down through the vertical spine until Blue is at $\varepsilon_{13}$. Red moves to $\theta_1$, Blue to $\theta_3$, and Red to $\theta_4$ (Red avoids the Blue goal node). They continue until Red reaches $\omega_{11}$. Blue moves to $\omega_{12}$ (choosing not to go to $\omega_{10}$, which would cut off access to goal nodes), then Red to $\omega_{13}$, Blue to $\omega_{14}$, Red to $\lambda_{15}$, and so on until Blue is at $\lambda_8$. From there, Red moves to $\lambda_6$, Blue to $\lambda_5$, Red to $\lambda_4$, Blue to $\lambda_3$, Red to $\lambda_2$, and Blue has no legal move left except Red's goal node.

Following the green path: from $\omega_6$, Blue moves to $\theta_{14}$, Red to $\theta_{13}$, Blue to $\theta_{12}$, Red to $\theta_{11}$, Blue to $\theta_{10}$.  From there, Red avoids $\theta_9$ (which would be a loss) and moves to $\theta_8$. Blue moves to $\varepsilon_{13}$, Red to $\theta_1$, Blue to $\theta_3$, Red to $\theta_4$, and the play continues until Red reaches $\omega_{11}$. From here the game proceeds identically to the red path, again forcing Blue onto Red's goal node.

Both branches of Red's response result in Blue losing, confirming that Blue's illegal deviation at $\varepsilon_8$ is successfully punished.

\section{Extension to Eight-Directional Movement}\label{sec:eight-directional}

An issue arises when extending our construction from the four-directional 
to the eight-directional setting: simply allowing diagonal moves on the same 
grid embedding used in Section~\ref{sec:gadgets} breaks the construction, 
since it lets a player bypass the very structure that enforces correct game 
logic. Consider the diode gadget shown in Figure~\ref{fig:diode-gadget}. Under 
cardinal (four-directional) movement, play is forced along the sequence 
$\mathit{start}_{\varepsilon} \to \varepsilon_0 \to \varepsilon_1 \to 
\varepsilon_2 \to \varepsilon_3$, with Blue starting, Red moving to 
$\varepsilon_0$, Blue to $\varepsilon_1$, Red to $\varepsilon_2$, and Blue to 
$\varepsilon_3$. If diagonal moves are permitted on this same embedding, 
however, Red can instead move directly from $\varepsilon_3$ to Red's goal 
node, which lies diagonally adjacent to it, short-circuiting the gadget and 
winning in a single move that the case analysis of 
Section~\ref{sec:gadgets} never accounts for. Establishing PSPACE-completeness for 
the eight-directional variant therefore requires an embedding in which such shortcuts are not available anywhere in the construction.

Noticing that the gadgets described in 
Section~\ref{sec:gadgets} are all designed using only orthogonal connections, we resolve this by transforming the pattern such that all orthogonal edges become diagonal edges, and all the nodes can be reached by the exact same neighboring nodes as before rotation.

To illustrate this concretely, Figure~\ref{fig:eight-dots-rotation}, on the left, shows a 
pattern of eight dots landing on eight squares. Each dot can be reached by any neighboring dots using either a horizontal or vertical path. By rotating those same eight dots 45 degrees and enlarging the distance slightly by $\sqrt{2}$ (the length of the diagonal of a square), while keeping the original grid 
intact, we have the figure on the right. This creates a new pattern in which one node can only be reached by the other using a diagonal path. Thus, we can embed the orthogonal game as one with only diagonal moves available, so the 8-directional variant will have no ways to ``cheat'' on it. This rotate-and-scale technique is analogous to the one used by  \cite{brunner2023} for chess for similar purposes.


\begin{figure}[h!]
    \centering
    \includegraphics[width=1\linewidth]{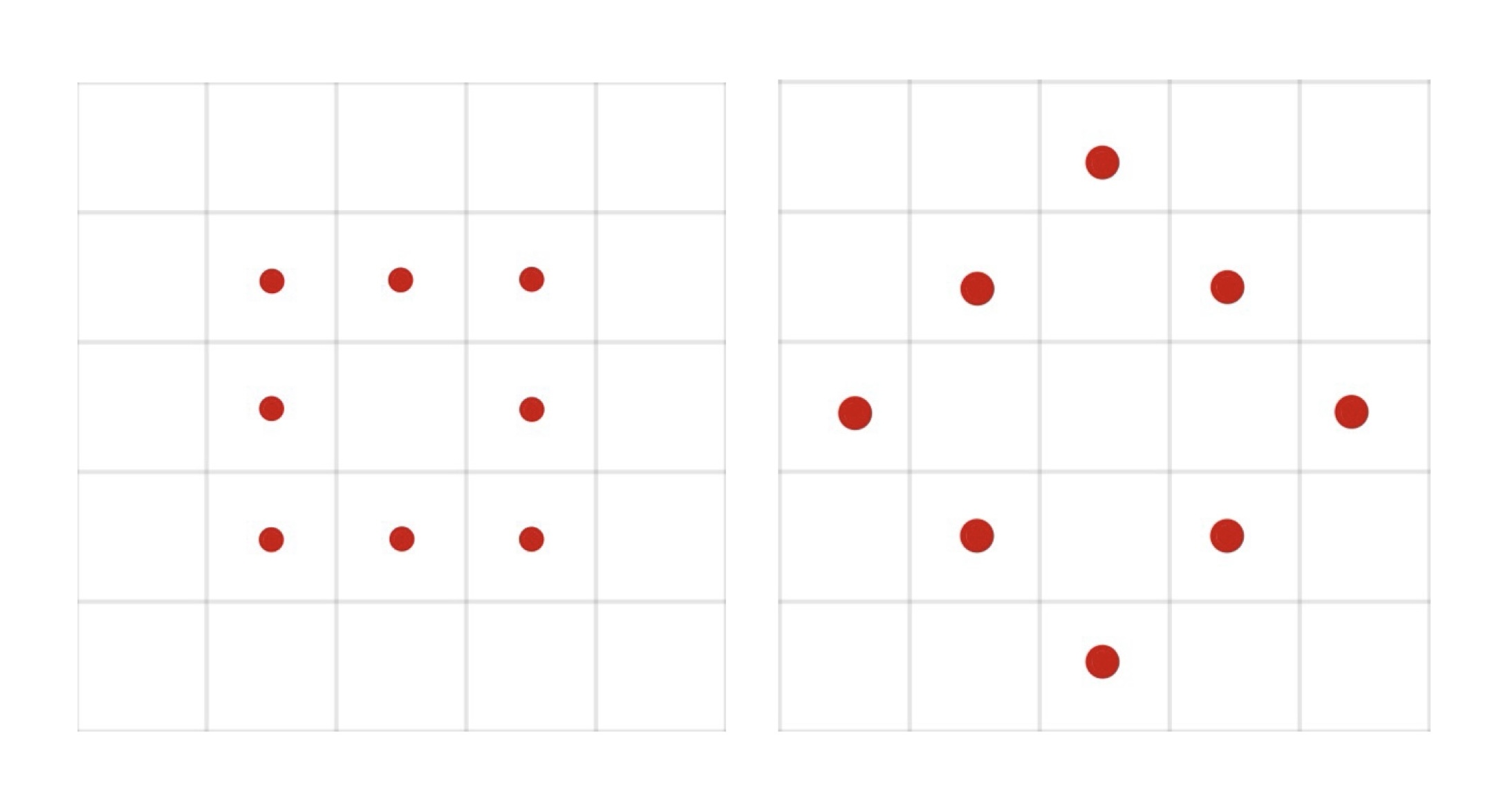}
    \caption{An eight-vertex pattern before (left) and after (right) a 
    45-degree rotation and slight enlargement. Every orthogonal edge in the 
    original pattern becomes a diagonal edge in the rotated pattern, and the 
    connected dots remain reachable from one another, while no new diagonal 
    shortcuts are introduced between previously unconnected dots.}
    \label{fig:eight-dots-rotation}
\end{figure}

We now write this out carefully. Let $G=(V,E)$ be the graph from 
Section~\ref{sec:gadgets}, drawn on $\mathbb{Z}^2$ so that two nodes $u,v$ are 
connected exactly when they are one step apart horizontally or vertically 
(distance $1$). If we have a vertex $v$ at $(x,y)$ then after the transformation $T$ we get the coordinates:
\[
T(v) = (x-y,\ x+y).
\]
This map rotates everything by 45 degrees and stretches it by a factor of 
$\sqrt2$. Let $V' = T(V)$ be the new set of node positions, with all unused squares becoming pre-slimed vertices, and build a new graph $G' = (V', E')$, where $E'$ is all possible 8 directional connections.

\begin{theorem}
Let $F(G)$ be the function that outputs which player wins the 4 directional game, and similarly $E(G)$ be the function for the 8 directional version. Then, $F(G) = E(G')$.
\end{theorem}

We'll complete the proof of this theorem after creating two helper lemmas.

\begin{lemma}
    Only diagonal moves are possible on $G'$.
    \end{lemma}
\begin{proof}
    Take any node $v=(x,y)$ in $V$. Its new position is $T(v)=(x-y,\ x+y)$. Notice 
    that the two coordinates of $T(v)$ always add up to an even number, since 
    $(x-y)+(x+y)=2x$. A straight (horizontal or vertical) move changes only one 
    coordinate by 1, which changes this sum by an odd number. Since they aren't a part of the original map, they must be a pre-slimed vertex. So a straight move can  never connect two points that both have an even coordinate sum. This means no two nodes in $V'$ can be connected by a straight move. In other words, only diagonal moves are 
    possible in $G'$.
\end{proof}
\begin{lemma}
    For $v = (x,y) \in G$, with $v_u = (x, y + 1)$, $v_r = (x+1, y)$, $v_d = (x, y-1)$, and $v_l = (x-1, y)$ being the adjacent vertices, $T(v)$ is below and to the right of $T(v_u)$, below and to the left of $T(v_r)$, above and to the left of $T(v_d)$, and above and to the right of $T(v_\ell)$ in $G'$
\end{lemma}
\begin{proof}
     We will refer to $T(v)$ as $v'$, $T(v_u)$ as $v'_u$, $T(v_r)$ as $v'_r$, $T(v_d)$ as $v'_d$, and $T(v_\ell)$ as $v'_l$. We get coordinates $v' = (x - y, x+ y)$, $v'_u = (x - y - 1, x + y + 1)$, $v'_r = (x - y + 1, x + y + 1)$, $v'_d = (x - y + 1, x + y - 1)$, and $v'_\ell = (x - y - 1, x + y - 1)$. This puts $v'_u$ as above and to the left of $v'$, $v'_r$ as above and to the right of $v'$, $v'_d$ as below and to the right of $v'$, and $v'_\ell$ and below and to the left of $v'$. As such, these edges are in $E$.
\end{proof}

\smallskip
Now we complete the proof of theorem 4:

\begin{proof}
The winning player on the 4 directional game has a winning strategy on the 8 directional game. By Lemma 6, we define a move transformation function $T_m$ which takes orthogonal movement and turns them into the corresponding digonal ones as follows:
$$T_m(1,0)=(1,1), \quad T_m(-1,0)=(-1,-1), \quad T_m(0,1)=(-1,1), \quad T_m(0,-1)=(1,-1).$$

Now, the winning player has the following strategy: take the winning move on the 4-directional game, and apply the transformation to get the move to the 8 directional one, and make that. Since $G'$ doesn't have any orthogonal moves (by Lemma 5), whatever the losing player's response is on the 8 directional one, we can apply the inverse transformation to map it to the 4 directional game, after which the winning player can again apply their winning response. This will continue to apply inductively, until a goal node is reached, which must be of the winning player's color.



\end{proof}

Therefore, the rotated and rescaled construction is a valid eight-directional 
grid game whose logical behavior is identical to the cardinal version, and 
the same PSPACE-completeness result holds. Thus, Slime Trail is 
PSPACE-complete on grids for both the four-directional and eight-directional variants.
\section{Conclusion}\label{sec:conclusion}

We have proved that Cardinal Grid Slime Trail is PSPACE-complete by adapting 
the QBF reduction described in \cite{ferland2017} to work within the geometric 
constraints of the integer lattice. Our grid-compatible gadgets demonstrate 
that the structural restrictions imposed by grid graphs do not reduce the 
computational complexity of Slime Trail.

The polynomial-time constructibility of the reduction follows from the fact 
that the entire construction has $O(m \cdot n)$ vertices and edges (as 
established in Section~\ref{sec:reduction}), and each gadget can be placed on 
the grid in time proportional to its size.

As shown in Section~\ref{sec:eight-directional}, the construction also extends, under a 45-degree rotation, to the eight-directional variant of the game. This resolves the open problem posed in \cite{ferland2017} and confirms that the game as actually played in 
competition, Slime Trail with eight-directional movement on a square grid as 
featured in the Portuguese National Mathematical Games Championship (CNJM) 
organized by Ludus \cite{ludus_cnjm, burke2017blog}, is computationally 
intractable.

\subsection{Future Work}

Several questions remain open.

\begin{enumerate}
    \item Does PSPACE-completeness extend to hexagonal grids, which have 
    different parity and degree properties?
    
    \item How does the complexity change if each player has only a single goal 
    node rather than multiple?
    
    \item Can the grid structure be exploited to develop exact algorithms that 
    perform better than worst-case exponential time for small instances, despite 
    PSPACE-completeness in general?
    
    \item Are there natural subclasses of Cardinal Grid Slime Trail instances---for 
    example, those with bounded treewidth---that admit efficient solutions?
\end{enumerate}

\appendix
\section{Complete Crossover Gadget Case Analysis}

This appendix comprises 18 illustrations depicting a comprehensive case analysis of all 18 instances in which an illegal move is initiated within the crossover gadget. For each such deviation, we demonstrate that the opposing player can either force a win or steer the game back onto its correct trajectory.

\begin{figure}[p]
    \centering
    \includegraphics[width=1\textwidth]{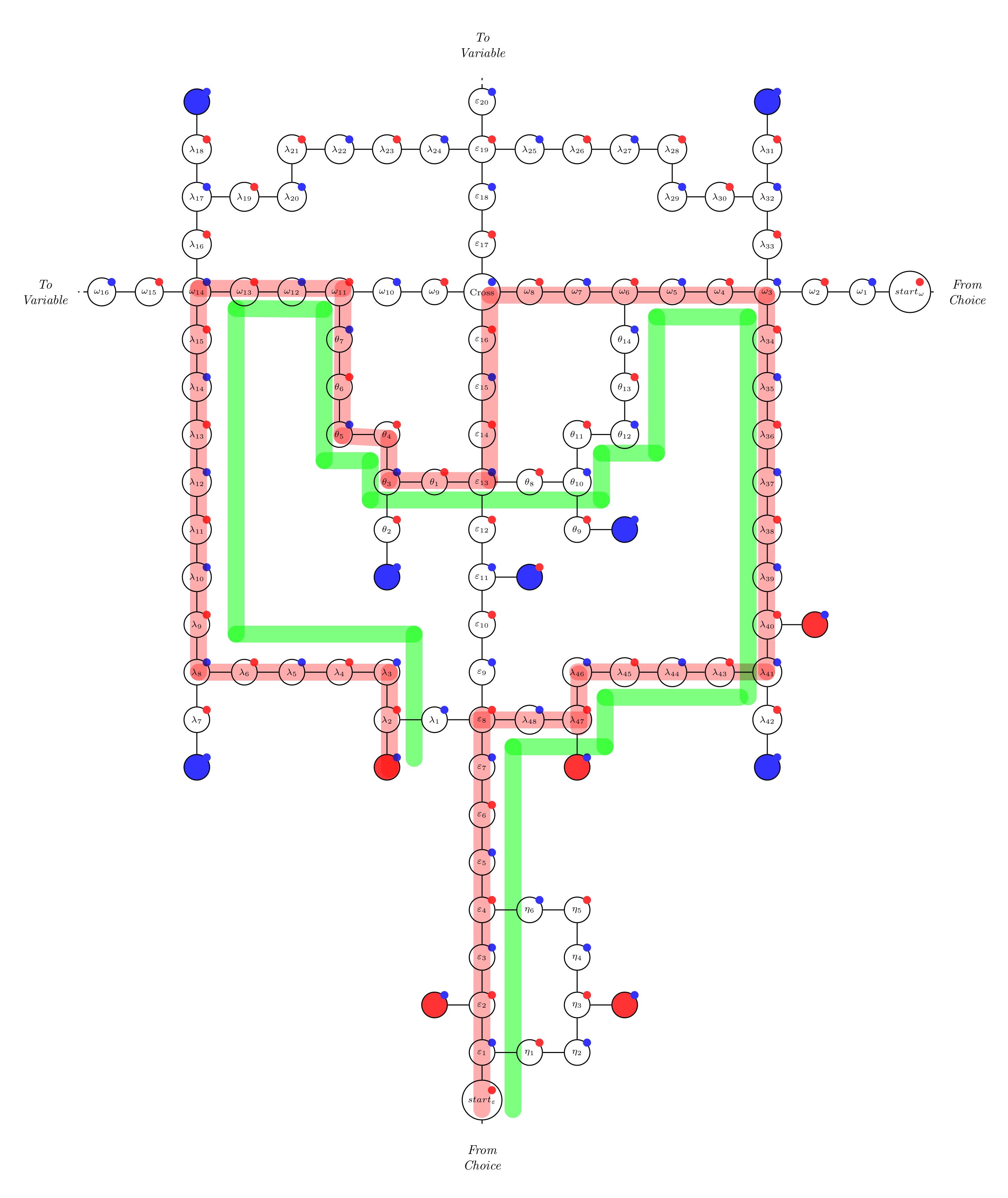}
    \caption{Case 1: Blue deviating at $\varepsilon_8$ by moving to $\lambda_{48}$.}
    \label{fig:appendix-case-1}
\end{figure}

\begin{figure}[p]
    \centering
    \includegraphics[width=1\textwidth]{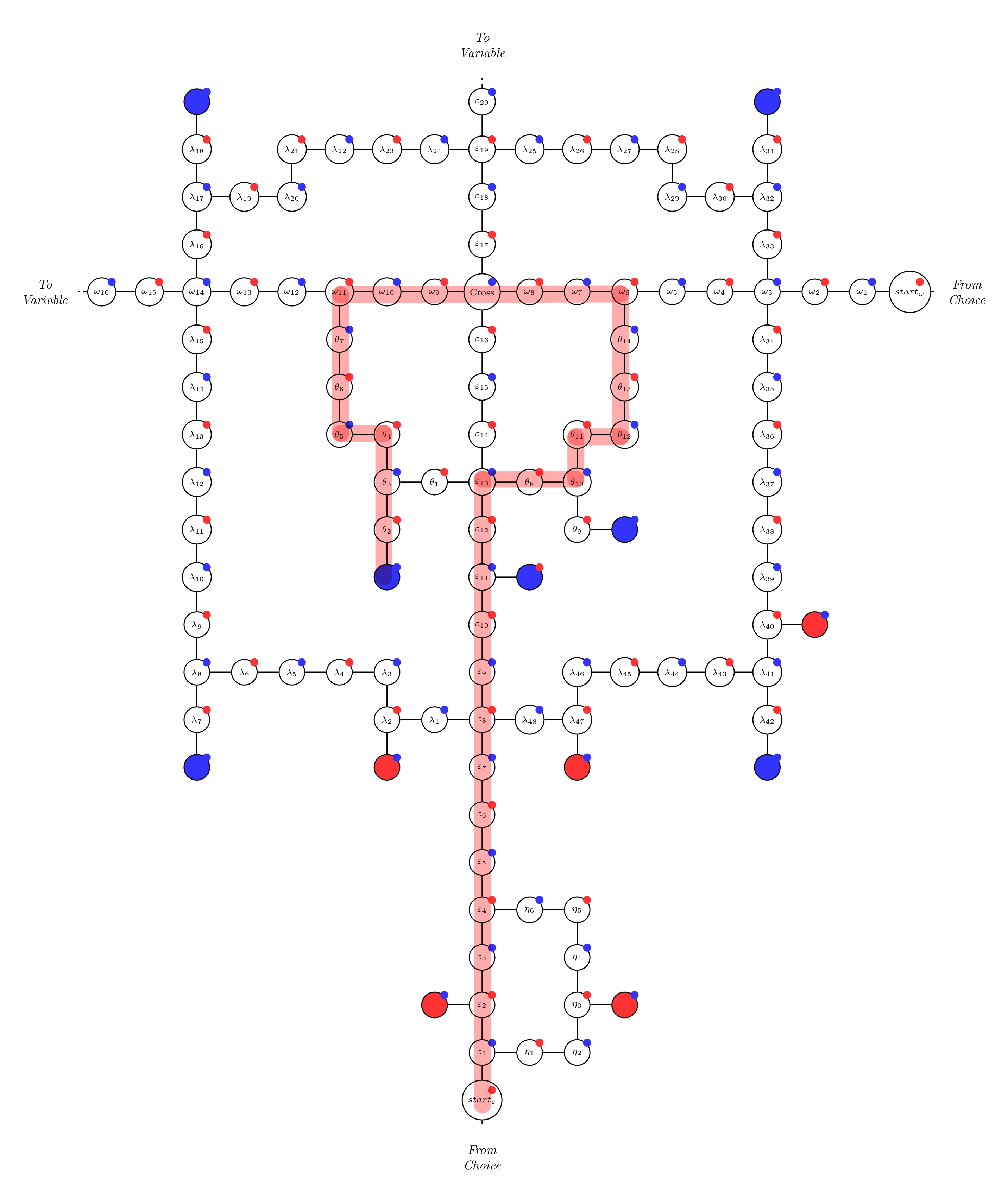}
    \caption{Case 2: Red deviating at $\varepsilon_{13}$ by moving to $\theta_8$.}
    \label{fig:appendix-case-2}
\end{figure}

\begin{figure}[p]
    \centering
    \includegraphics[width=1\textwidth]{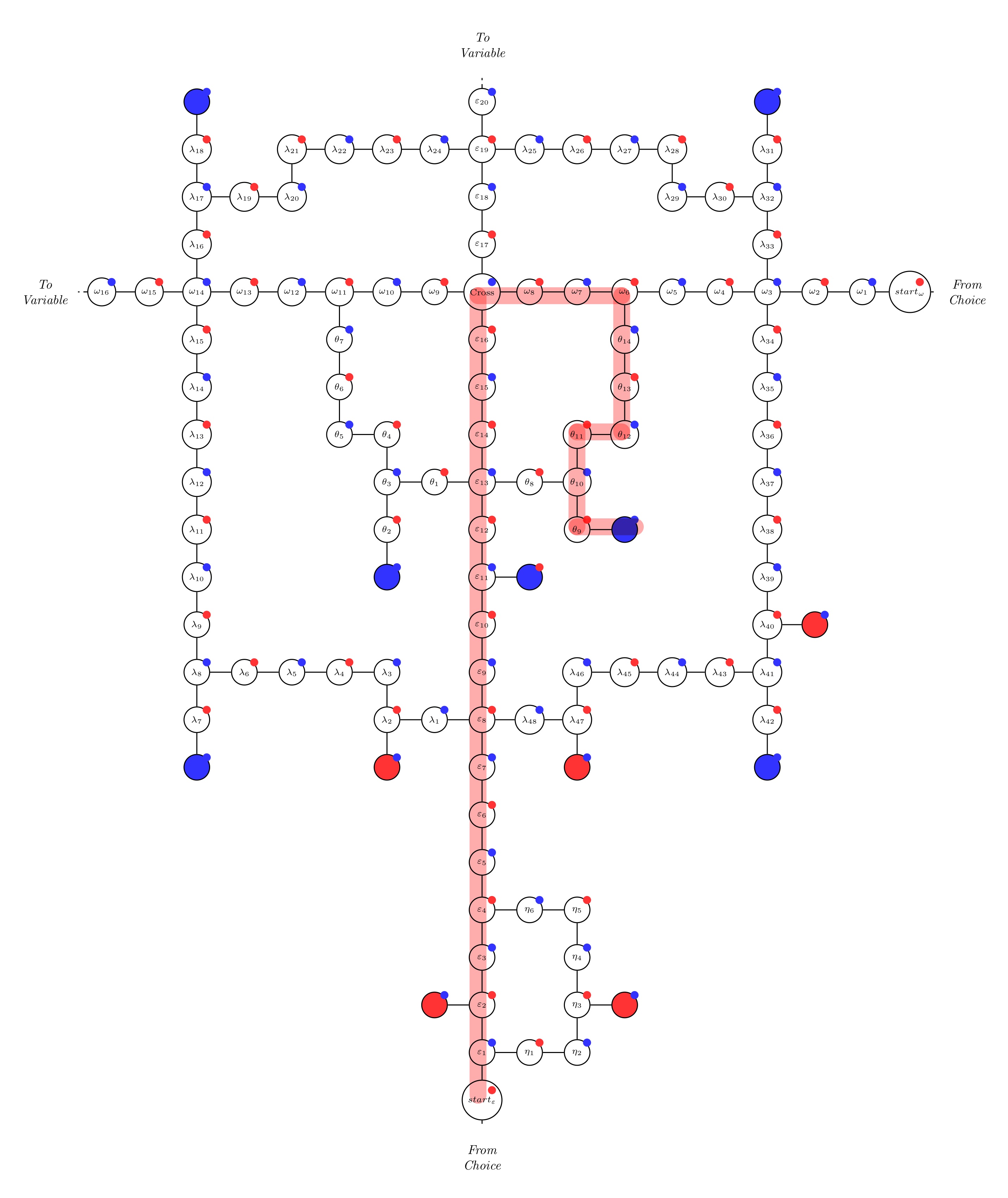}
    \caption{Case 3: Red deviating at Cross by moving to $\omega_8$.}
    \label{fig:appendix-case-3}
\end{figure}

\begin{figure}[p]
    \centering
    \includegraphics[width=1\textwidth]{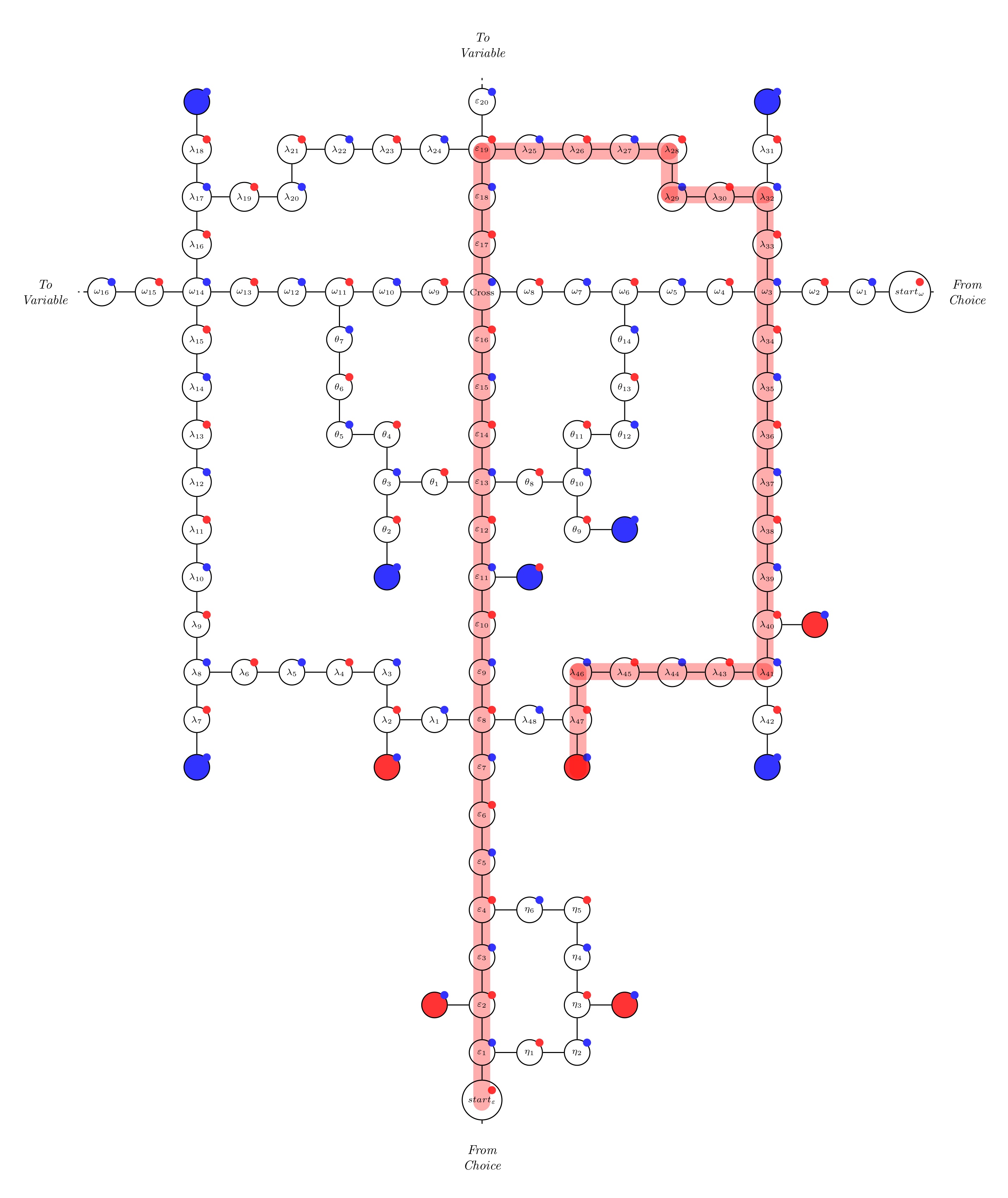}
    \caption{Case 4: Blue deviating at $\varepsilon_{19}$ by moving to $\lambda_{25}$.}
    \label{fig:appendix-case-4}
\end{figure}

\begin{figure}[p]
    \centering
    \includegraphics[width=1\textwidth]{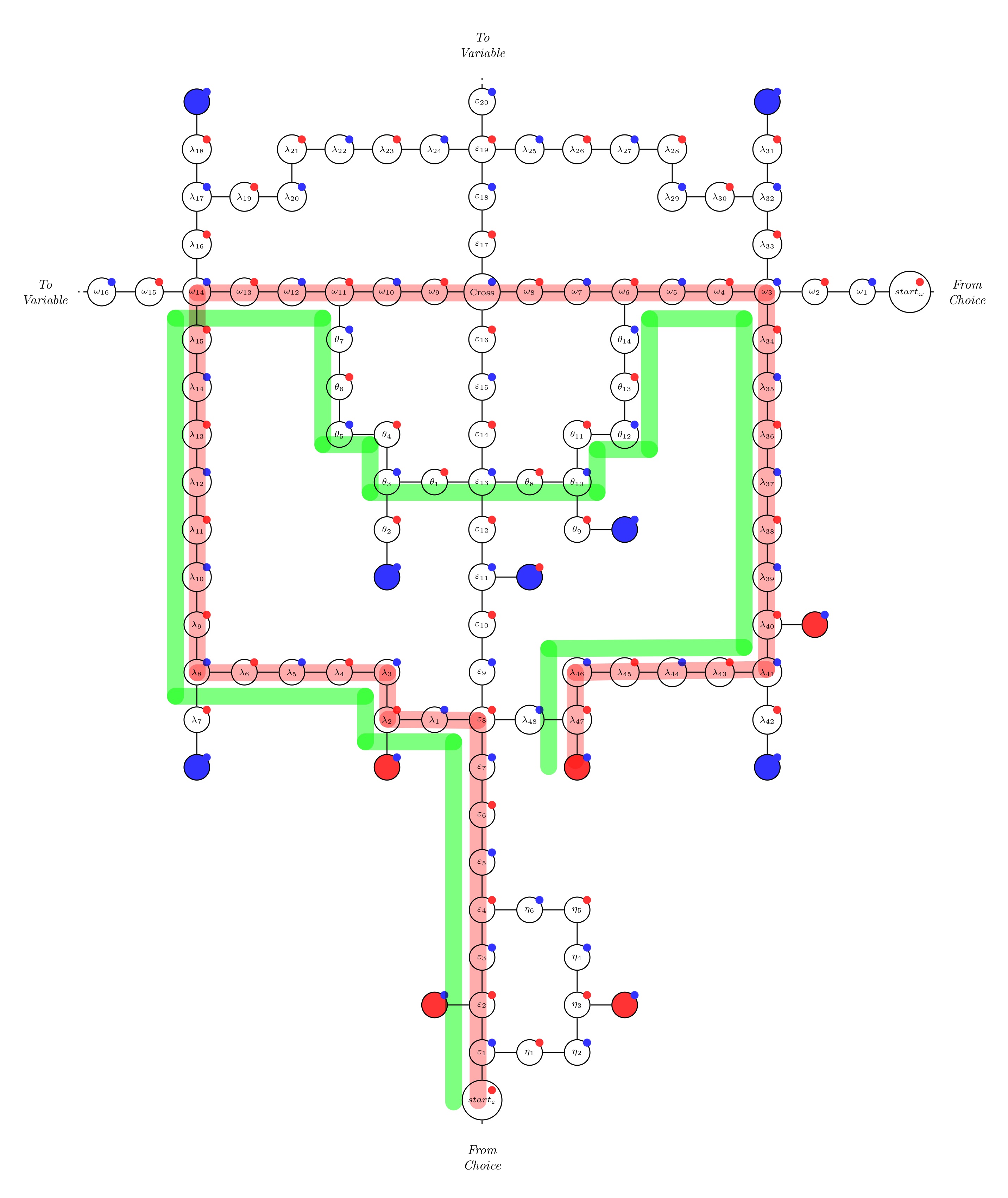}
    \caption{Case 5: Blue deviating at $\varepsilon_8$ by moving to $\lambda_1$.}
    \label{fig:appendix-case-5}
\end{figure}

\begin{figure}[p]
    \centering
    \includegraphics[width=1\textwidth]{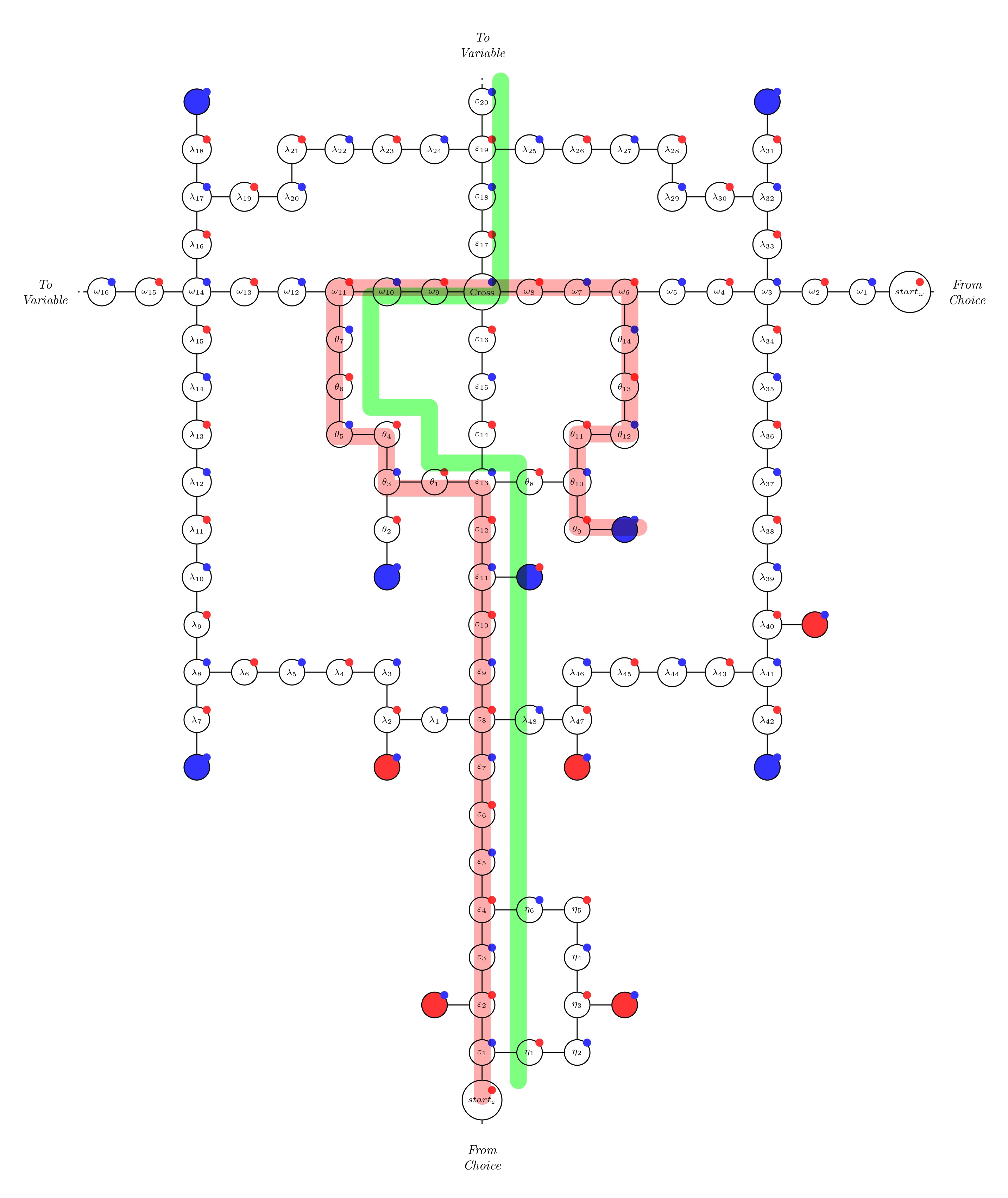}
    \caption{Case 6: Red deviating at $\varepsilon_{13}$ by moving to $\theta_1$.}
    \label{fig:appendix-case-6}
\end{figure}

\begin{figure}[p]
    \centering
    \includegraphics[width=1\textwidth]{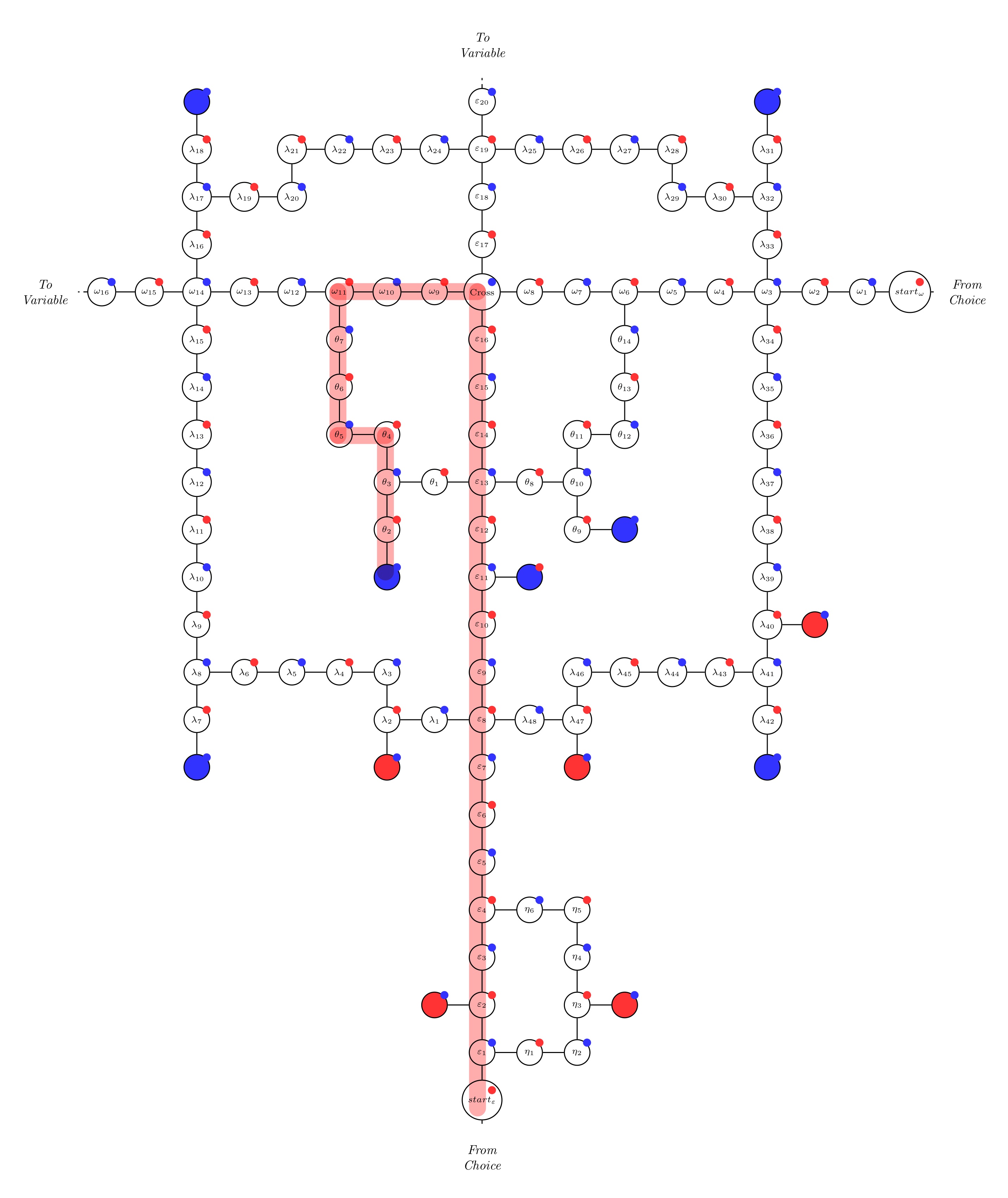}
    \caption{Case 7: Red deviating at Cross by moving to $\omega_9$.}
    \label{fig:appendix-case-7}
\end{figure}

\begin{figure}[p]
    \centering
    \includegraphics[width=1\textwidth]{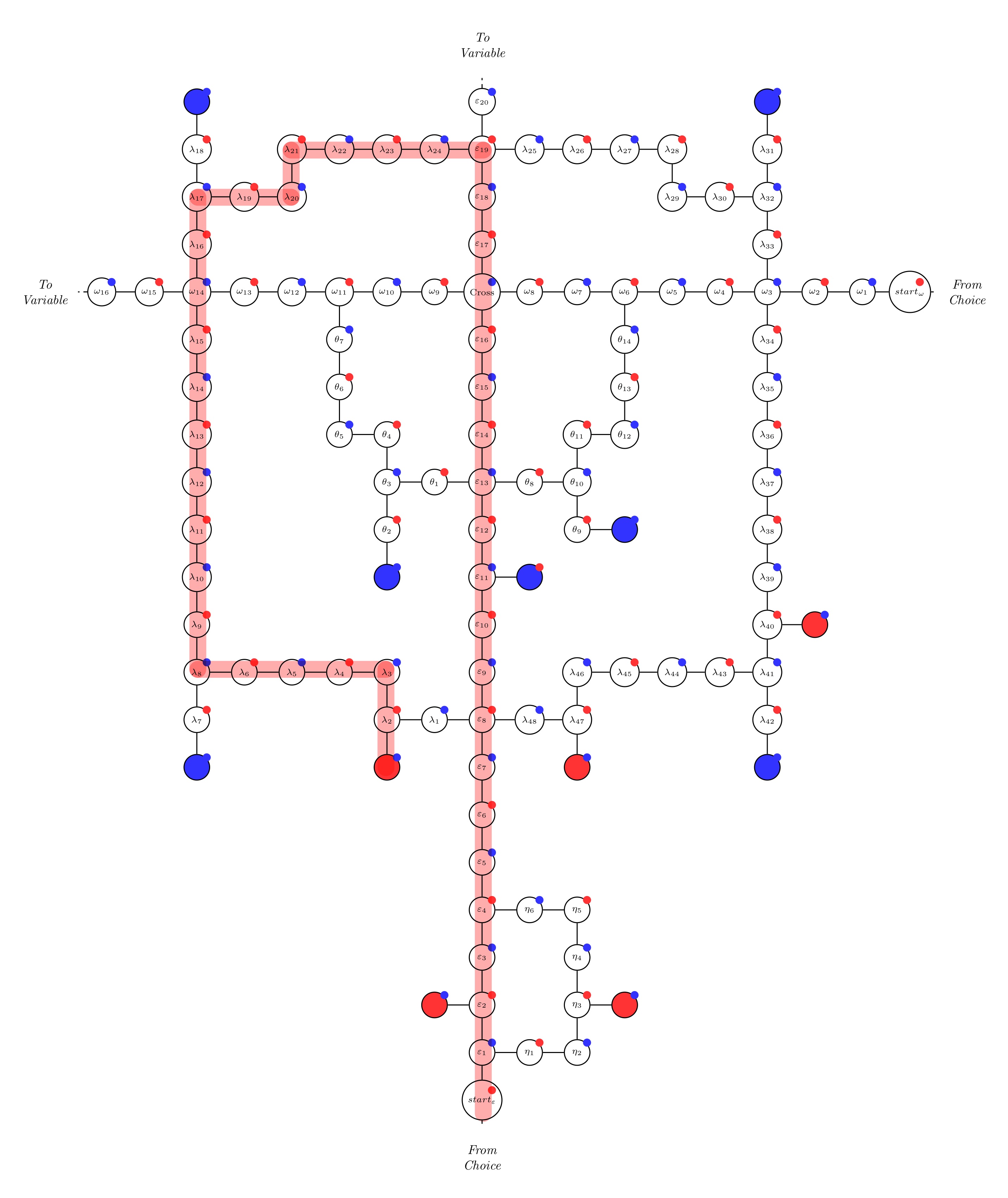}
    \caption{Case 8: Blue deviating at $\varepsilon_{19}$ by moving to $\lambda_{24}$.}
    \label{fig:appendix-case-8}
\end{figure}

\begin{figure}[p]
    \centering
    \includegraphics[width=1\textwidth]{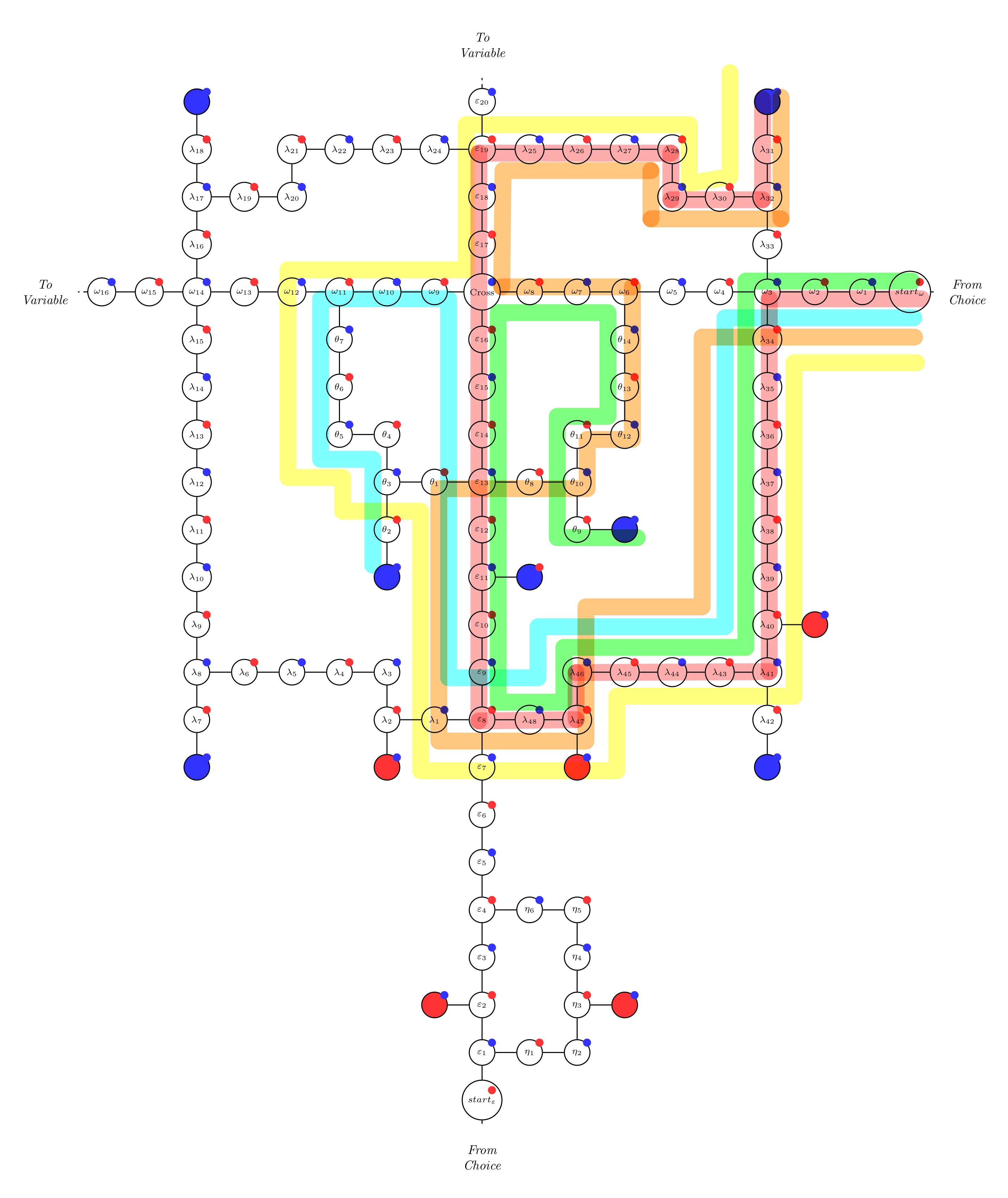}
    \caption{Case 9: Red deviating at $\omega_3$ by moving to $\lambda_{34}$.}
    \label{fig:appendix-case-9}
\end{figure}

\begin{figure}[p]
    \centering
    \includegraphics[width=1\textwidth]{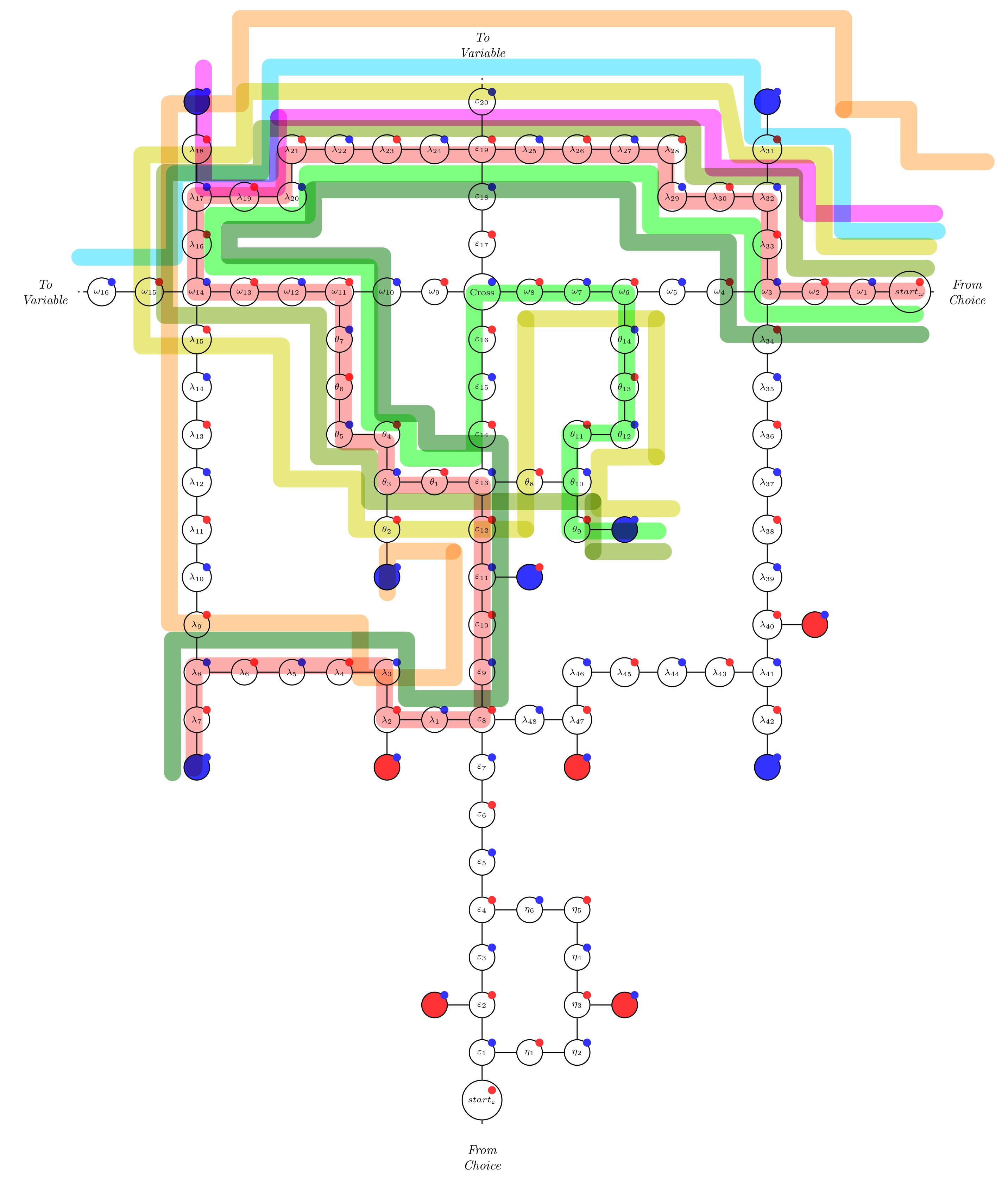}
    \caption{Case 10: Red deviating at $\omega_3$ by moving to $\lambda_{33}$.}
    \label{fig:appendix-case-10}
\end{figure}

\begin{figure}[p]
    \centering
    \includegraphics[width=1\textwidth]{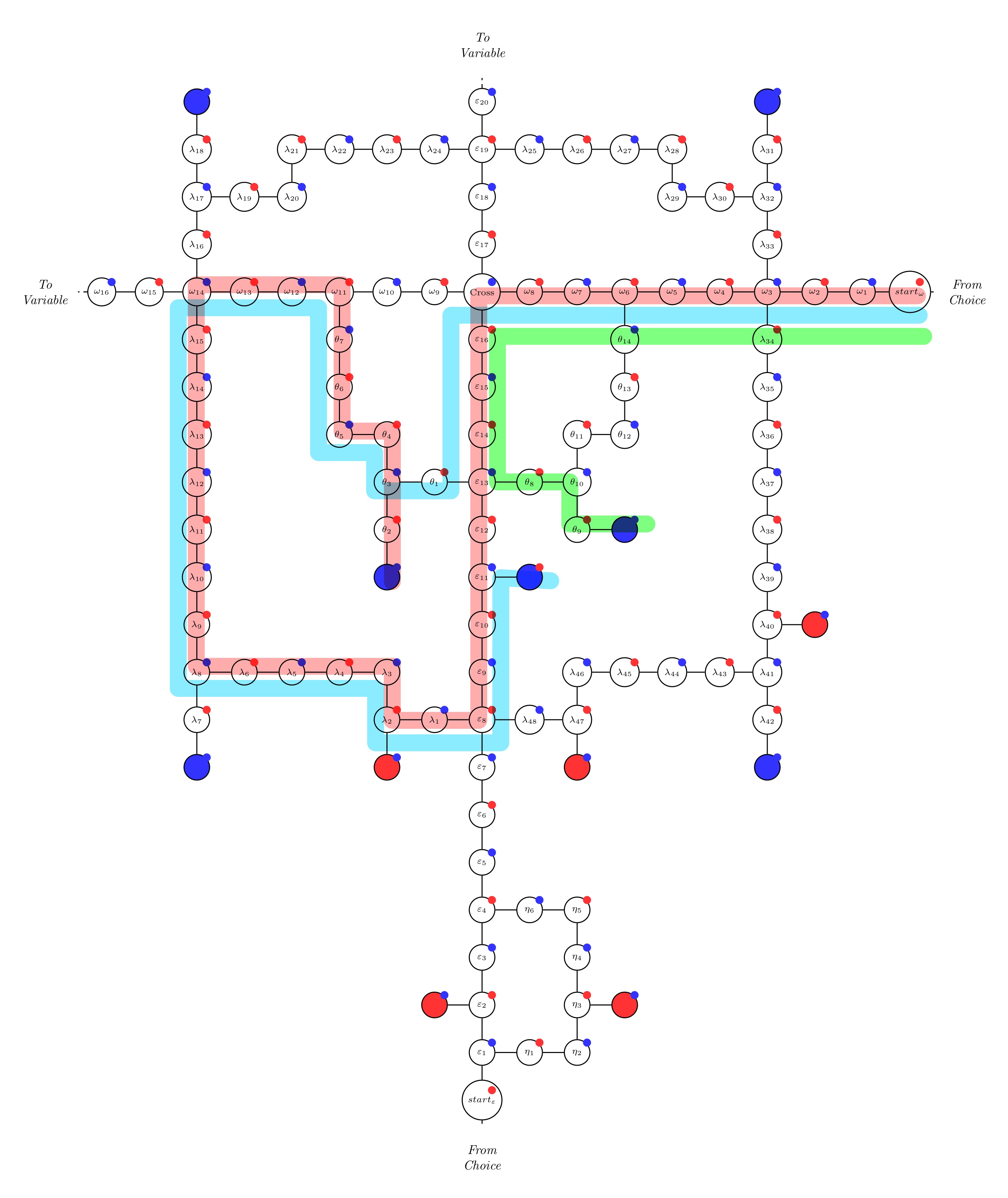}
    \caption{Case 11: Red deviating at Cross by moving to $\varepsilon_{16}$.}
    \label{fig:appendix-case-11}
\end{figure}

\begin{figure}[p]
    \centering
    \includegraphics[width=1\textwidth]{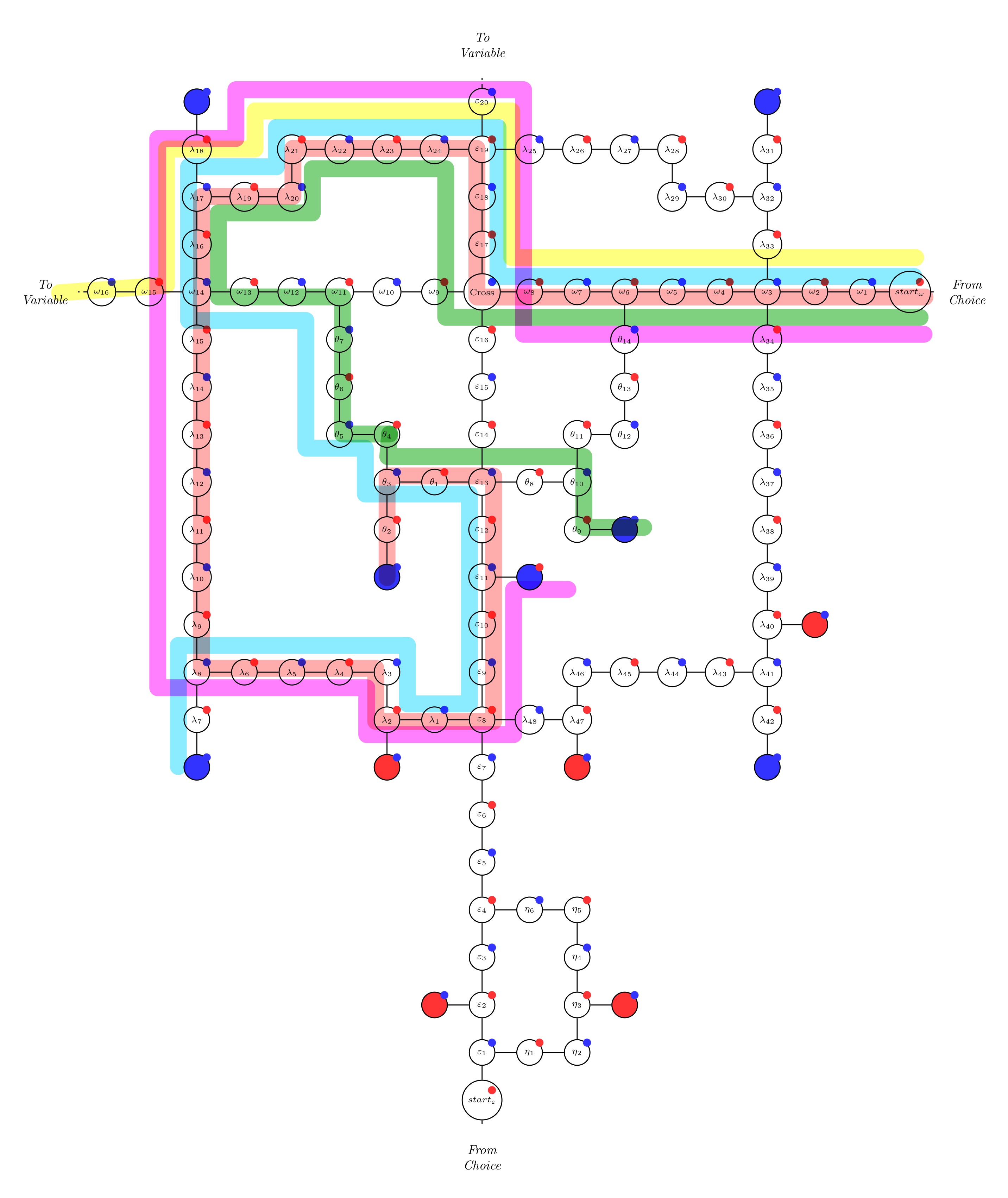}
    \caption{Case 12: Red deviating at Cross by moving to $\varepsilon_{17}$.}
    \label{fig:appendix-case-12}
\end{figure}

\begin{figure}[p]
    \centering
    \includegraphics[width=1\textwidth]{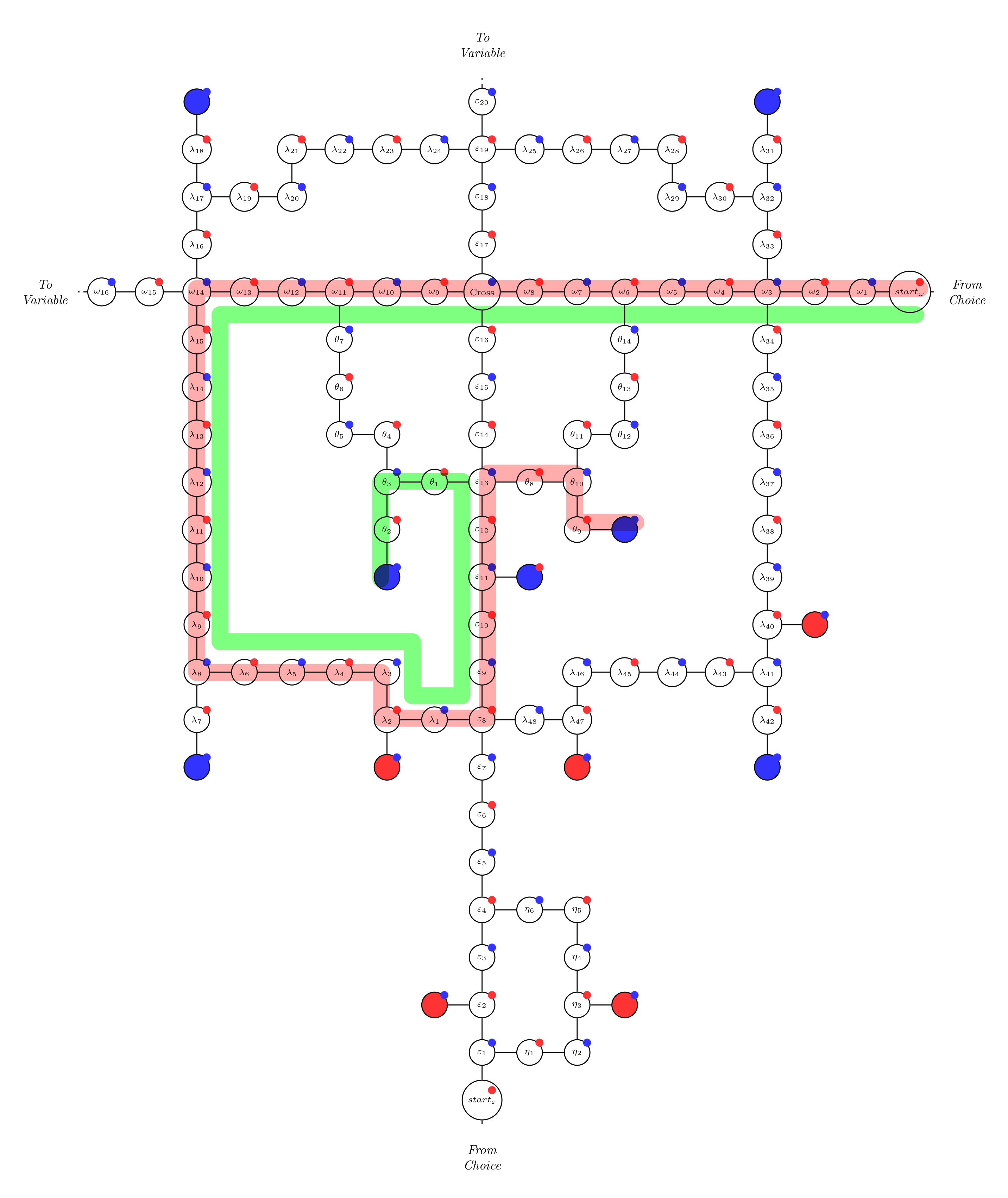}
    \caption{Case 13: Red deviating at $\omega_{14}$ by moving to $\lambda_{15}$.}
    \label{fig:appendix-case-13}
\end{figure}

\begin{figure}[p]
    \centering
    \includegraphics[width=1\textwidth]{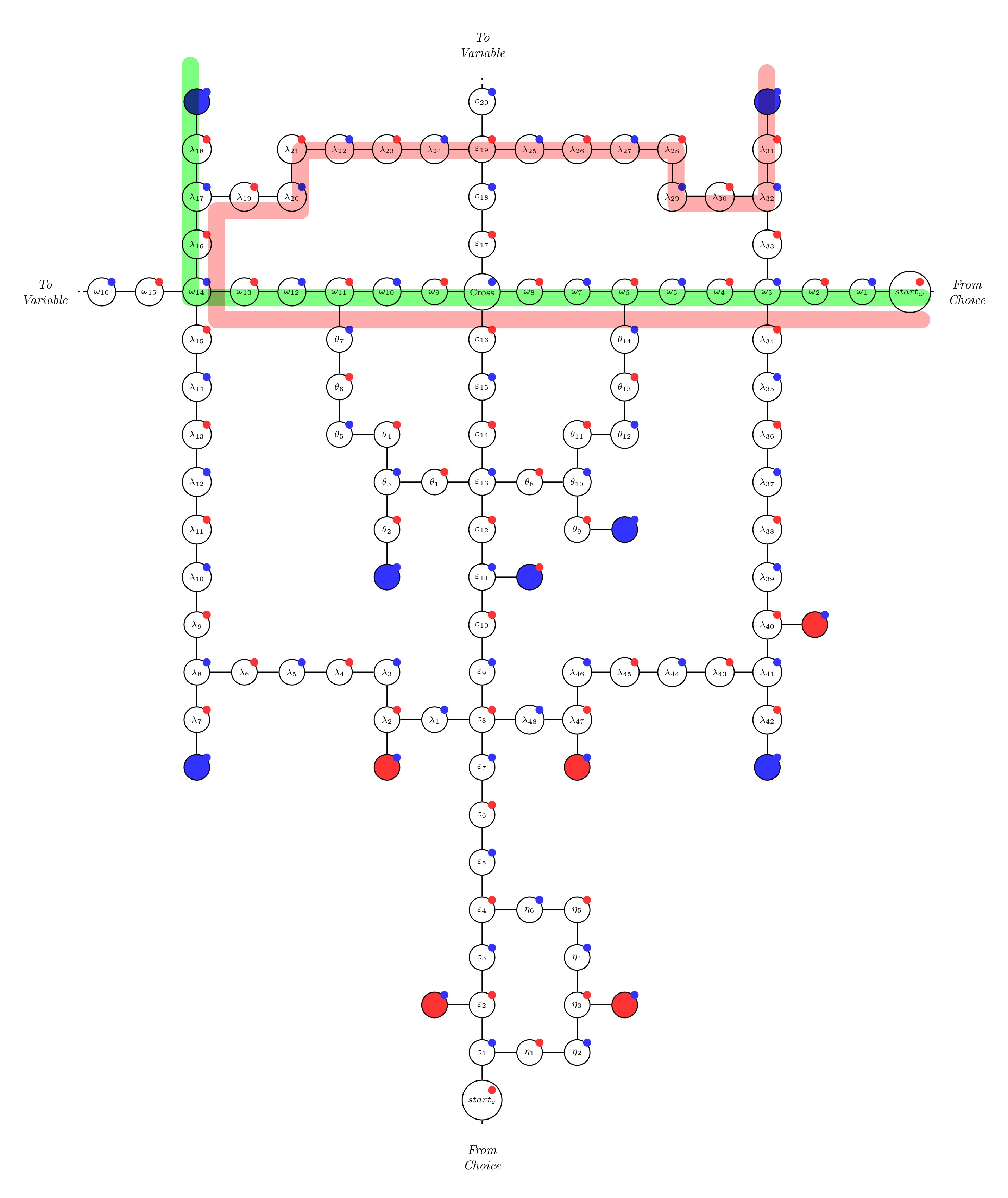}
    \caption{Case 14: Red deviating at $\omega_{14}$ by moving to $\lambda_{16}$.}
    \label{fig:appendix-case-14}
\end{figure}

\begin{figure}[p]
    \centering
    \includegraphics[width=1\textwidth]{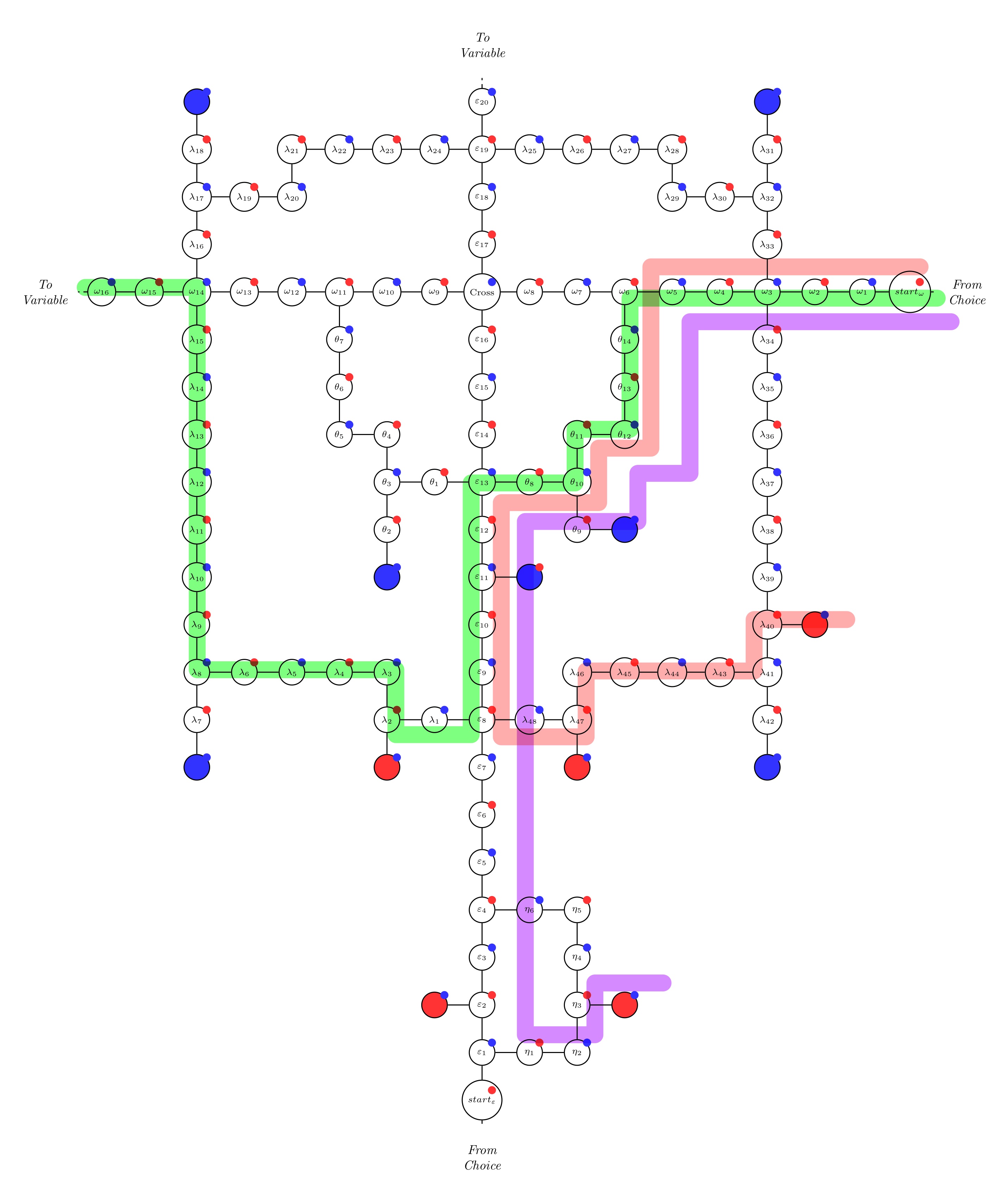}
    \caption{Case 15: Blue deviating at $\omega_6$ by moving to $\theta_{14}$.}
    \label{fig:appendix-case-15}
\end{figure}

\begin{figure}[p]
    \centering
    \includegraphics[width=1\textwidth]{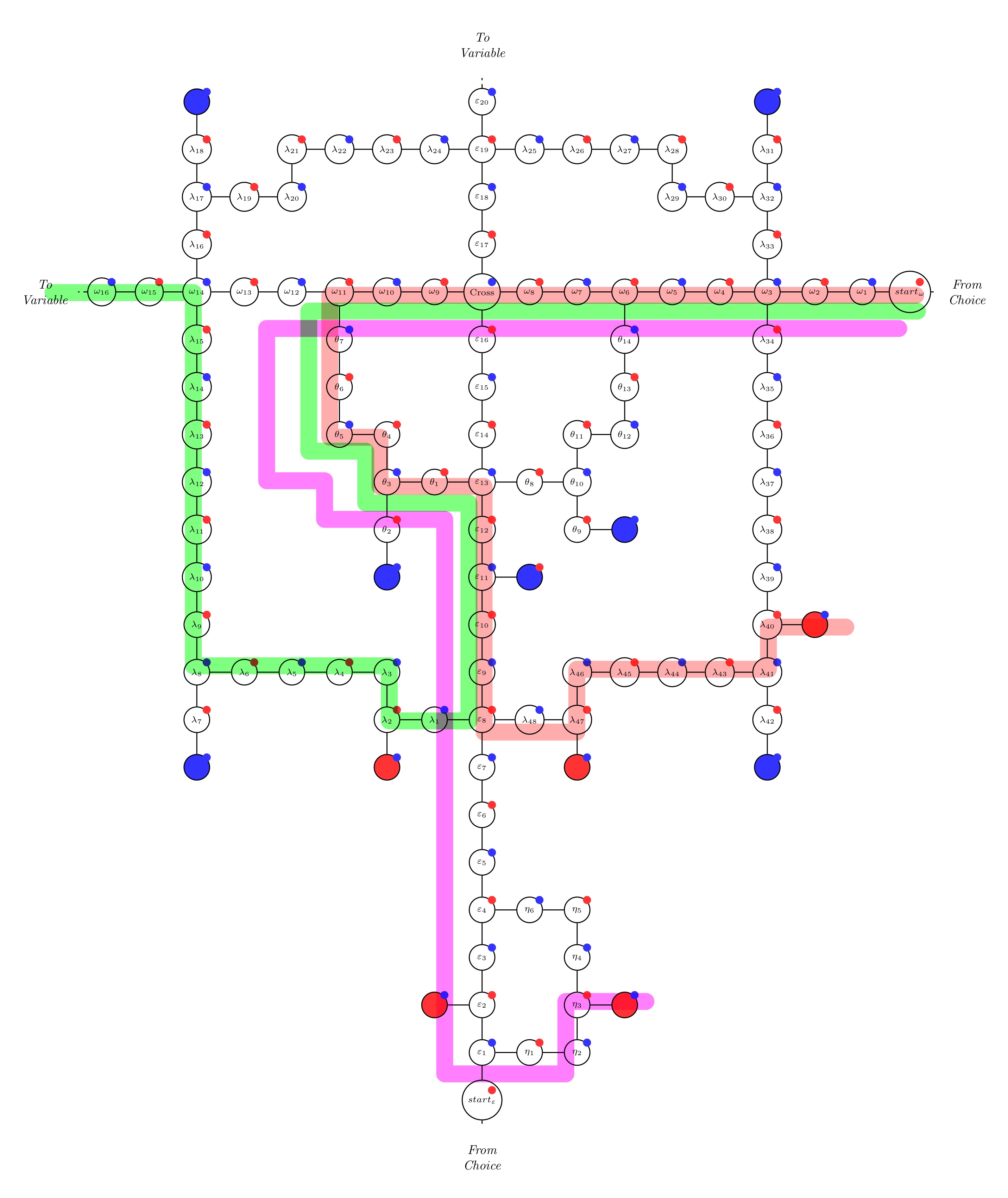}
    \caption{Case 16: Blue deviating at $\omega_{11}$ by moving to $\theta_7$.}
    \label{fig:appendix-case-16}
\end{figure}

\begin{figure}[p]
    \centering
    \includegraphics[width=1\textwidth]{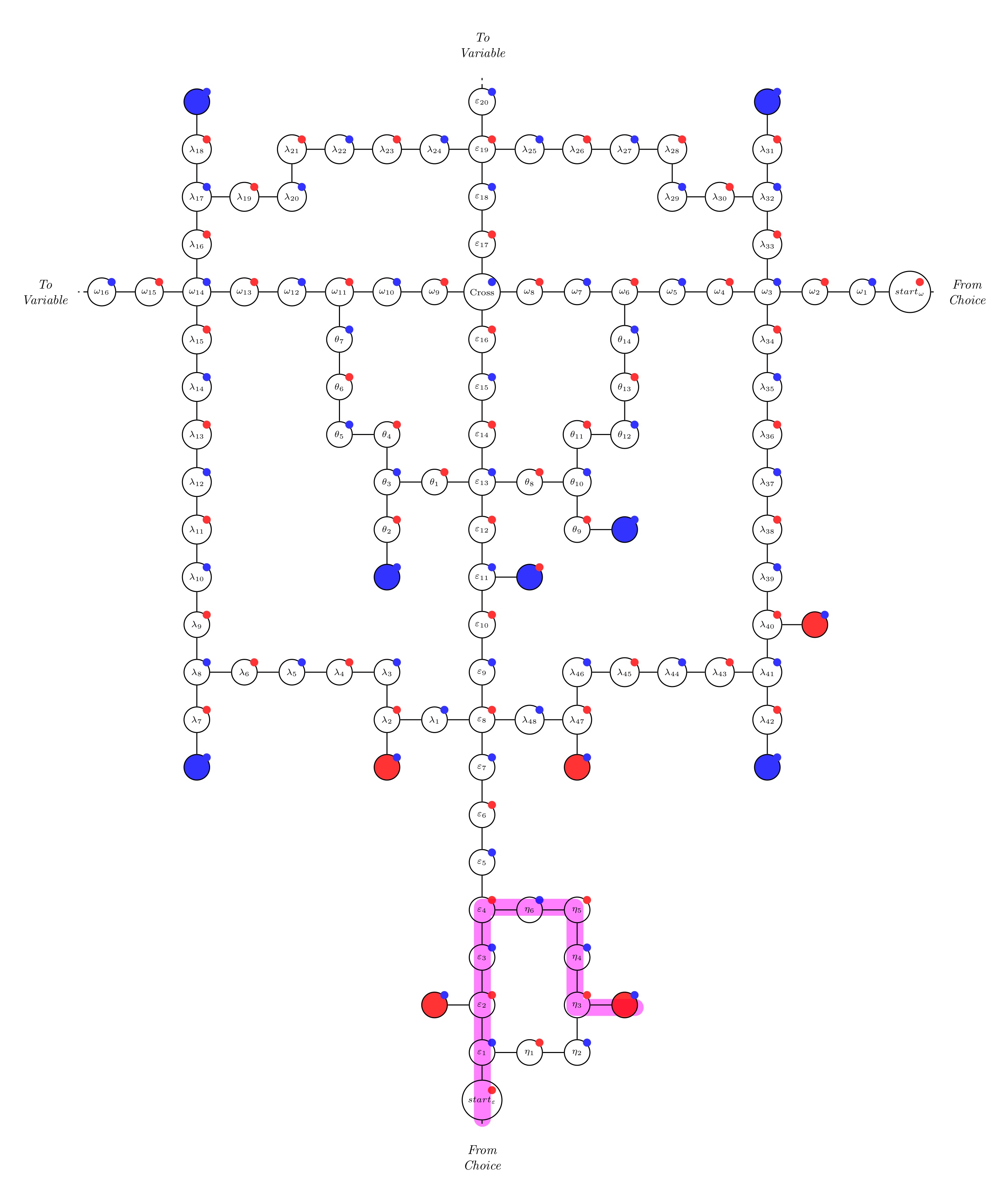}
    \caption{Case 17: Blue deviating at $\varepsilon_4$ by moving to $\eta_6$.}
    \label{fig:appendix-case-17}
\end{figure}

\begin{figure}[p]
    \centering
    \includegraphics[width=1\textwidth]{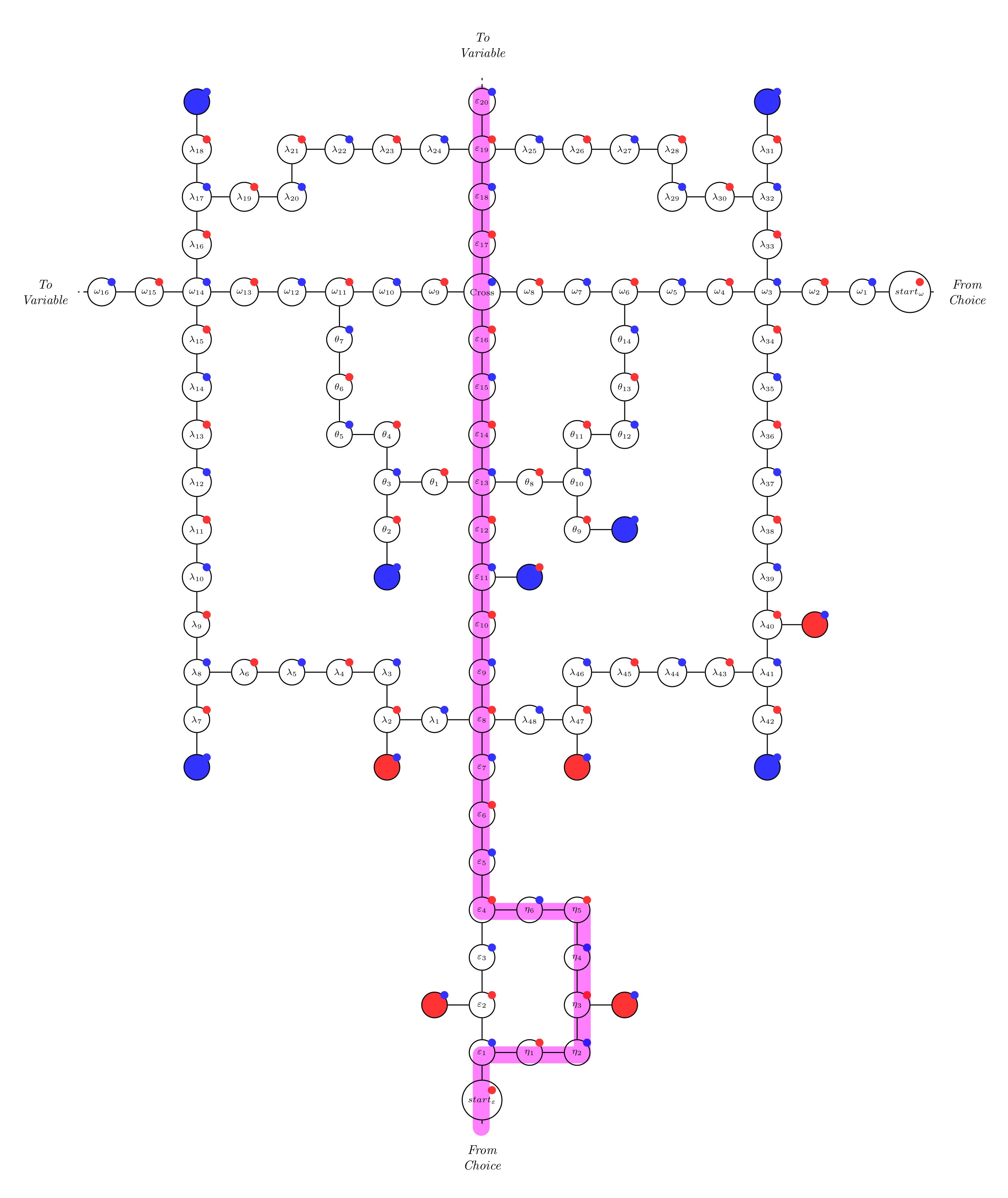}
    \caption{Case 18: Red deviating at $\varepsilon_1$ by moving to $\eta_1$.}
    \label{fig:appendix-case-18}
\end{figure}

\clearpage
\bibliographystyle{plain}

\end{document}